\documentclass[12pt,onecolumn]{IEEEtran}
\usepackage{amsmath}
\usepackage{epsfig}
\usepackage{epsf}
\usepackage{amsmath}
\usepackage{amssymb}
\usepackage{color}
\usepackage{cases}
\usepackage{epic,eepic}
\usepackage{graphicx}

\newtheorem{corollary}{Corollary}
\newtheorem{definition}{Definition}  
\newtheorem{theorem}{Theorem}		
\newtheorem{obs}{Observation} 
\newtheorem{prop}{Proposition}
\newtheorem{lemma}{Lemma}

\newcommand\piund{{\pi}}

\newcommand\qc	{{\cal Q}}
\newcommand\scc	{{\cal S}}

\newcommand\xc	{{\cal X}}
\newcommand\yc		{{\cal Y}}

\newcommand\fc  	{{\cal F}}

\newcommand\oneb	{{\boldsymbol{1}}}

\newcommand\Wb	{{\boldsymbol{W}}}
\newcommand\Xb	{{\boldsymbol{X}}}

\newcommand\Zb  	{{\boldsymbol{Z}}}

\newcommand\ep	{{\varepsilon}}
\newcommand\de  	{{\delta}}

\newcommand\code {\mbox{${\textstyle~{\cal C}\!\!\!\!\!\!\!\sim}$}}
\newcommand\codeb{{~{\cal C}\!\!\!\!\!\!\sim}}

\begin{document}

\newcommand{\ddoublespace}{\renewcommand{\baselinestretch}{1.6}\small\normalsize}
\newcommand{\Doublespace}{\renewcommand{\baselinestretch}{1.5}\small\normalsize}

\ddoublespace

\begin{center}
\begin{tabular}[b]{c}
{\LARGE \bf Mathematics and Engineering} \\ \medskip
{\LARGE \bf Communications Laboratory}
\ \\
\ \\
{\Large \sf Technical Report}
\end{tabular}\hspace{0.4in}\epsfig{file=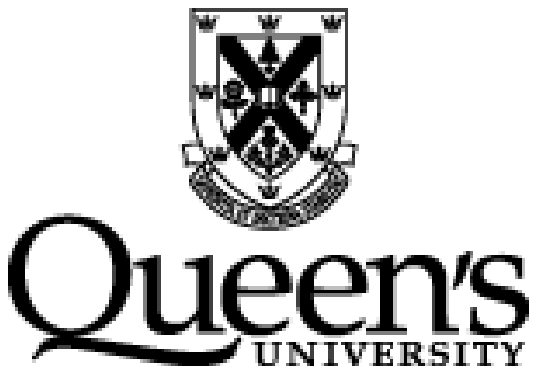,width=3.0cm}

\rule{6.0in}{0.8mm}
\end{center}

\bigskip
\bigskip
\bigskip
\bigskip
\vspace{1.0in}

\begin{center}
{\Large\bf A Generalized Poor-Verd\'{u} Error Bound for Multihypothesis} \\
\smallskip
{\Large\bf Testing and the Channel Reliability Function} \\
\vspace{0.6in}
{\large {\it P.-N.~Chen} and {\it F.~Alajaji}}  \
\
\end{center}

\vspace{2.0in}
\begin{center}
{\large July 2011}
\end{center}

\newpage

\Doublespace

\title{A Generalized Poor-Verd\'{u} Error Bound for Multihypothesis Testing
and the Channel Reliability Function}

\author{Po-Ning~Chen and Fady Alajaji}

\maketitle

\begin{abstract}
A lower bound on the minimum error probability for multihypothesis testing
is established. The bound, which is expressed in terms of the   
cumulative distribution function of the {\it tilted} posterior 
hypothesis distribution given the observation with tilting parameter $\theta\ge 1$,
generalizes an earlier bound due the Poor and Verd\'{u} (1995).
A sufficient condition is established under which the new bound (minus a multiplicative
factor) provides the exact error probability in the limit of $\theta$ going to infinity. Examples illustrating the new bound are also provided.

The application of this generalized Poor-Verd\'{u} bound to the channel reliability
function is next carried out, resulting in two information-spectrum upper bounds.
It is observed that, for a class of channels including 
the finite-input memoryless Gaussian channel, one of the bounds is tight and
gives a multi-letter asymptotic expression for the reliability function, 
albeit its determination or calculation in single-letter form remains an open
challenging problem. Numerical examples regarding the other bound are finally 
presented. 
\end{abstract}

\begin{keywords}
Hypothesis testing, probability of error, maximum-a-posteriori and maximum likelihood estimation, 
channel coding, channel reliability function, error exponent, 
binary-input additive white Gaussian noise channel.
\end{keywords}

\renewcommand{\thefootnote}{}
\footnotetext{
P.~N.~Chen is with the Department of Electrical Engineering (Institute of Communications Engineering), 
National Chiao Tung University, Taiwan, ROC (email: poning@faculty.nctu.edu.tw).
F.~Alajaji is with the Department of Mathematics and Statistics,
Queen's University, Kingston, Ontario K7L 3N6, Canada
(email: fady@mast.queensu.ca).

This work was supported in part by 
the Natural Sciences and Engineering Research Council
of Canada and the National Science Council of
Taiwan, R.O.C.
}
\renewcommand{\thefootnote}{\arabic{footnote}}
\setcounter{footnote}{0}

\newpage
\section{Introduction}\label{intro}
In \cite{PV95}, Poor and Verd\'{u} establish a lower bound to the minimum error probability
of multihypothesis testing. Specifically, given two random variables 
$X$ and $Y$ with joint distribution $P_{X,Y}$, $X$ taking values in a finite
or countably-infinite alphabet $\xc$ and $Y$ taking values in an arbitrary alphabet $\yc$, they
show that the optimal maximum-a-posteriori (MAP) estimation of $X$ given $Y$ results in
the following lower bound on the probability of estimation error $P_e$: 
\begin{equation}
P_e \ge (1-\alpha)P_{X,Y}\left\{(x,y)\in\xc\times\yc:~
P_{X|Y}(x|y)\leq \alpha\right\}
\label{pv-bound}
\end{equation}
for each $\alpha\in[0,1]$, where $P_{X|Y}$ denotes the posterior distribution of $X$
given $Y$ and the prior distribution $P_X$ is arbitrary (not necessarily
uniform).
This bound has pertinent information-theoretic applications such as in the proof of
the converse part of the channel coding theorem that yield formulas for both capacity 
and $\ep$-capacity for general channels with memory 
(not necessarily information stable, stationary, etc) \cite{VH94,PV95}. It also 
improves upon previous lower bounds
due to Shannon \cite{shannon57}, \cite[Eq.~(7)]{PV95} 
and to Verd\'{u} and Han \cite{VH94}, \cite[Eq.~(9)]{PV95}.

Furthermore, Poor and Verd\'{u} use the above bound
to establish an information-spectrum based upper bound to the
reliability function $E^\ast(R)$ -- i.e., the optimal error exponent or
the largest rate of asymptotic exponential
decay of the error probability of channel 
codes \cite{gallager,blahut,csiszar,viterbi}--
of general channels \cite[Eq.~(14)]{PV95}.
They conjecture that this bound, which is expressed in terms of a large-deviation rate function for
the normalized channel information density (see Section~\ref{prelim} for the definition),  
is tight (i.e., exactly equal to $E^\ast(R)$) for all rates $R$.
In \cite{pv02}, it is however shown via a counterexample involving the memoryless
binary erasure channel (BEC) that the bound is not tight at low rates, and a slightly
tighter bound is presented \cite[Corollary~1]{pv02}.

In this work, we generalize the above Poor-Verd\'{u} lower bound in \eqref{pv-bound}
for the minimum error probability of multihypothesis testing. The new bound is expressed
in terms of the cdf of the tilted posterior distribution of $X$ given $Y$ 
with tilting parameter $\theta \ge 1$, and it reduces to \eqref{pv-bound} when $\theta=1$;
see Theorem~\ref{improved-pv}. 
We also provide a sufficient condition under which our generalized Poor-Verd\'{u} bound,
without the multiplicative factor $(1-\alpha)$,
is exact in the limit of $\theta$ going to infinity. Specifically, the sufficient condition
requires having a unique MAP estimate
of $X$ from $Y$ almost surely in $P_Y$, where $P_Y$ is the distribution of $Y$;
see Theorem~\ref{tightness-theo}. We present a few examples to illustrate the results
of Theorems~\ref{improved-pv} and~\ref{tightness-theo}.  

We proceed by applying the above results to the reliability
function $E^\ast(R)$ of general channels.
We employ Theorem~\ref{improved-pv} to establish two information-spectrum 
upper bounds to $E^\ast(R)$; see Theorem~\ref{exp-upper-bound}.
One upper bound, $E_{\text{PV}}^{(\theta)}(R)$, is a function
of the tilting parameter $\theta$, while the other bound, $\bar{E}_{\text{PV}}(R)$, involves
taking the limit infimum of $\theta$. It turns out that if the channel satisfies
a symmetry condition, then both upper bounds can be expressed in terms of the information
density of an auxiliary channel whose transition distribution is nothing but the tilted 
distribution of the original channel distribution; see Observation~\ref{auxiliary-channel}.  

We next use Theorem~\ref{tightness-theo} to show that for the memoryless finite-input
additive white Gaussian noise (AWGN) channel, the upper bound $\bar{E}_{\text{PV}}(R)$
is tight, hence yielding an information-spectral formula
for this channel's reliability function:
$E^\ast(R)=\bar{E}_{\text{PV}}(R)$ for all rates $R$ between 0 and channel capacity; 
see Theorem~\ref{awgn-exact-exp-theo}.
The calculation or determination in closed (single-letter) form
of $\bar{E}_{\text{PV}}(R)$ is however a formidable task and remains 
a notoriously open problem, 
as it requires solving the optimization
of a large-deviation rate function in additions to two limiting operations; this makes it
quite difficult to compare $\bar{E}_{\text{PV}}(R)$ to well-known lower/upper bounds
to $E^\ast(R)$ (such as the random coding lower bound and the
sphere packing upper bound \cite{gallager,blahut}\footnote{The sphere
packing bound \cite{gallager}
is referred to as the space partitioning bound in \cite{blahut}.}) 
for this AWGN channel. 
Nevertheless, the above multi-letter asymptotic expression for $E^\ast(R)$
may be conceptually useful for the future determination
of $E^\ast(R)$ in computable single-letter form at low rates.\footnote{For the finite-input
AWGN channel as well as the whole class of memoryless channels, $E^\ast(R)$ is already
exactly determined in terms of a simple (single-letter) expression at high rates
(beyond some critical rate) since the random coding and sphere-packing bounds
coincide in that rate region \cite{gallager}. Further improvements were recently
established for the memoryless binary symmetric channel (BSC) and the continuous-input
AWGN channel in \cite{barg,ben-haim}, where it is shown that 
$E^\ast(R)$ is also exactly determined for rates $R$
in some interval directly below the critical rate.} 
We also note that the equality $E^\ast(R)=\bar{E}_{\text{PV}}(R)$ holds
for a class of channels satisfying the sufficient condition of Theorem~\ref{tightness-theo};
see Corollary~\ref{exact-exp-cor} and Observation~\ref{cont-channel-obs}.

Finally, we provide a lower bound to $E_{\text{PV}}^{(\theta)}(R)$ for the case of memoryless
channels, which is computable for a given value of $\theta$. We use this lower bound
to demonstrate numerically that for the memoryless BSC,
$E_{\text{PV}}^{(\theta)}(R)$ is not tight at all rates when $\theta=1$ (which corresponds to the
original Poor-Verd\'{u} reliability function upper bound). We also
numerically show that for the memoryless Z-channel, $E_{\text{PV}}^{(\theta)}(R)$ is 
not tight at high rates for all considered values of $\theta$ (including large ones).

The rest of the paper is organized as follows. In Section~\ref{gen-pv-bound-dec},
the generalized Poor-Verd\'{u} lower bound to the multihypothesis testing 
minimum error probability is established in terms of the tilted
posterior distribution with parameter $\theta$ (Theorem~\ref{improved-pv}). 
A sufficient condition under which
an exact expression for the error probability is given in terms of an asymptotic (in $\theta$)
term of the bound (minus a multiplying factor) is also shown (Theorem~\ref{tightness-theo}).
Examples illustrating Theorems~\ref{improved-pv} and~\ref{tightness-theo} are provided
in Section~\ref{examples}. In Section~\ref{channel-reliability},
the two upper bounds, given by $E_{\text{PV}}^{(\theta)}(R)$ and $\bar{E}_{\text{PV}}(R)$, respectively, 
for the channel reliability function are proved (Theorem~\ref{exp-upper-bound}).
Furthermore, it is noted that $\bar{E}_{\text{PV}}(R)$ provides an exact asymptotic
characterization for the channel reliability function at all rates for the finite-input
AWGN channel as well as other channels (Theorem~\ref{awgn-exact-exp-theo} and Corollary~\ref{exact-exp-cor}).
Numerical examples involving the BSC and the Z-channel indicating the looseness
of $E_{\text{PV}}^{(\theta)}(R)$ for specific choices of $\theta$ are next provided.   
Finally, conclusions are stated in section~\ref{conclusion}.
Note that we will use the natural logarithm throughout.

\section{A generalized error lower bound for multihypothesis testing}\label{gen-pv-bound-dec}

We herein generalize the Poor-Verd\'{u} lower bound in (\ref{pv-bound})
for the multihypothesis testing error probability.

\bigskip

Consider two (correlated) random variables $X$ and $Y$, where
$X$ has a discrete (i.e., finite or countably infinite) alphabet
$\xc=\left\{x_1,x_2,x_3,\ldots\right\}$ and $Y$ takes on
values in an arbitrary alphabet $\yc$. The minimum 
probability of error $P_e$ in estimating $X$ from $Y$ is given by
\begin{equation}
P_e\triangleq\Pr\left[X\neq e(Y)\right]
\label{pe-definition}
\end{equation}
where $e(Y)$ is the MAP estimate defined as 
\begin{equation}
\label{eq:pv1}
e(Y)=\arg \max_{x\in\xc}P_{X|Y}(x|Y).
\end{equation}

\bigskip

\begin{theorem}\label{improved-pv}{\rm
The above minimum 
probability of error $P_e$ in estimating $X$ from $Y$ satisfies
the following inequality
\begin{equation}
P_e\geq(1-\alpha)P_{X,Y}\left\{(x,y)\in\xc\times\yc:~
P_{X|Y}^{(\theta)}(x|y)\leq \alpha\right\}
\label{gen-pv-bound}
\end{equation}
for each $\alpha\in[0,1]$ and $\theta\geq 1$, where for each $y \in \yc$,
\begin{equation}\label{twistdist}
P_{X|Y}^{(\theta)}(x|y)\triangleq \frac{P_{X|Y}^\theta(x|y)}
{\sum_{x'\in\xc}P_{X|Y}^\theta(x'|y)}, \quad x \in \xc,
\end{equation}
is the tilted distribution of $P_{X|Y}(\cdot|y)$ with parameter $\theta$
\cite{Bucklewbook}.
}\end{theorem}

\medskip
\noindent
{\it Note:} When $\theta=1$, the above bound in (\ref{gen-pv-bound}) reduces
to the  Poor-Verd\'{u} bound in (\ref{pv-bound}).

\medskip
\begin{proof} Fix $\theta\geq 1$. We only provide the proof for $\alpha<1$ since 
the lower bound trivially holds when $\alpha=1$.

From (\ref{pe-definition}) and (\ref{eq:pv1}), the minimum error probability $P_e$ 
incurred in testing among the
values of $X$ satisfies
\begin{eqnarray*}
1-P_e
&=& \Pr[X=e(Y)] \\
&=&\int_\yc P_{X|Y}(e(y)|y) \ dP_Y(y)\\
&=&\int_\yc
\left(\max_{x\in\xc}P_{X|Y}(x|y)\right)dP_Y(y)\\
&=&\int_\yc \left(\max_{x\in\xc}f_x(y)\right)dP_Y(y)\\
&=&E\left[\max_{x\in\xc}f_x(Y)\right],
\end{eqnarray*}
where $f_x(y)\triangleq P_{X|Y}(x|y)$.
For a fixed $y\in \yc$, let
$h_j(y)$ be the $j$-th element in the set
$$\{f_{x_1}(y),f_{x_2}(y),f_{x_3}(y),\ldots\}$$
such that its elements are listed in non-increasing order; i.e., 
$$h_1(y)
\geq h_2(y)\geq h_3(y)\geq \cdots$$ and
$$\{h_1(y),h_2(y),h_3(y),\ldots\}=
\{f_{x_1}(y),f_{x_2}(y),f_{x_3}(y),\ldots\}.$$
Then
\begin{equation}
\label{eq:firstpv}
1-P_e=E[h_1(Y)].
\end{equation}
Furthermore, for each $h_j(y)$ above, define $h_j^{(\theta)}(y)$ such that 
$h_j^{(\theta)}(y)$ be the respective element for $h_j(y)$ satisfying
$$h_j(y)=f_{x_j}(y)=P_{X|Y}(x_j|y)\ \Leftrightarrow\ h_j^{(\theta)}(y)=P_{X|Y}^{(\theta)}(x_j|y).$$
Since $h_1(y)$ is the largest among $\{h_j(y)\}_{j\geq 1}$, 
$$h_1^{(\theta)}(y)=\frac{h_1^\theta(y)}{\sum_{j\geq 1} h_j^\theta(y)}=\frac{1}{1+
\sum_{j\geq 2} [h_j(y)/h_1(y)]^\theta}$$ is non-decreasing in $\theta$ for each $y$; this implies that
\begin{equation}
h_1^{(\theta)}(y)\geq h_1(y) \quad \text{for $\theta\ge 1$ and $y\in\yc$}.
\label{monotone-prop}
\end{equation}

\noindent
For any $\alpha\in[0,1)$, we can write
$$P_{X,Y}\left\{(x,y)\in\xc\times\yc:~P_{X|Y}^{(\theta)}(x|y)>\alpha\right\}
=\int_\yc
P_{X|Y}\left\{x\in\xc:~P_{X|Y}^{(\theta)}(x|y)>\alpha\right\}dP_Y(y).$$
Noting that
\begin{eqnarray*}
P_{X|Y}\left\{x\in\xc~:~P_{X|Y}^{(\theta)}(x|y)>\alpha\right\}
&=&\sum_{x\in\xc}P_{X|Y}(x|y)\cdot\oneb\left(P_{X|Y}^{(\theta)}(x|y)>\alpha\right)\\
&=&\sum_{j=1}^\infty h_j(y)\cdot\oneb\left(h_j^{(\theta)}(y)>\alpha\right),
\end{eqnarray*}
where $\oneb(\cdot)$ is the indicator function, yields 
\begin{eqnarray}
P_{X,Y}\left\{(x,y)\in\xc\times\yc:~P_{X|Y}^{(\theta)}(x|y)>\alpha\right\}
&=&\int_\yc\left(\sum_{j=1}^\infty h_j(y)\cdot\oneb\left(h^{(\theta)}_j(y)>\alpha\right)\right)dP_Y(y) \nonumber\\
&\geq&\int_\yc h_1(y)\cdot\oneb(h^{(\theta)}_1(y)>\alpha) dP_Y(y) \nonumber \\
&\geq&\int_\yc h_1(y)\cdot\oneb(h_1(y)>\alpha) dP_Y(y) \nonumber \\
&=&E[h_1(Y)\cdot\oneb(h_1(Y)>\alpha)] \label{h-indicate},
\end{eqnarray}
where the second inequality follows from \eqref{monotone-prop}.
To complete the proof, we next 
relate $E[h_1(Y)\cdot\oneb(h_1(Y)>\alpha)]$
with $E[h_1(Y)]$, which is exactly $1-P_e$.
Invoking \cite[eq.~(19)]{PV95}, we have that 
for any $\alpha\in[0,1]$ and any random variable $U$ with $\Pr\{0\leq U\leq 1\}=1$, 
the following inequality holds with probability one
$$U\leq \alpha+(1-\alpha)\cdot U\cdot \oneb(U>\alpha).$$
Thus
$$E[U]\leq \alpha + (1-\alpha)E[U\cdot\oneb(U>\alpha)].$$
Applying the above inequality to (\ref{h-indicate})
by setting $U=h_1(Y)$, we obtain
\begin{eqnarray}
(1-\alpha)P_{X,Y}\left\{(x,y)\in\xc\times\yc~:~P_{X|Y}^{(\theta)}(x|y)>\alpha\right\}
&\geq&(1-\alpha)E[h_1(Y)\cdot\oneb(h_1(Y)>\alpha)] \nonumber \\
&\geq&E[h_1(Y)]-\alpha \nonumber\\
&=&(1-P_e)-\alpha \nonumber \\
&=&(1-\alpha)-P_e, \nonumber
\end{eqnarray}
where the first equality follows from \eqref{eq:firstpv}.
\end{proof}

We next show that if the MAP estimate $e(Y)$ of $X$ from $Y$ is almost surely
unique in (\ref{eq:pv1}), then
the bound of Theorem~\ref{improved-pv}, without the $(1-\alpha)$ factor, 
is tight in the limit of $\theta$ going to infinity.

\bigskip

\begin{theorem}\label{tightness-theo}{\rm
Consider two random variables $X$ and $Y$, where
$X$ has a finite or countably infinite alphabet $\xc=\left\{x_1,x_2,x_3,\ldots\right\}$
and $Y$ has an arbitrary alphabet $\yc$. Assume that
\begin{equation}\label{condition}
P_{X|Y}(e(y)|y)>\max_{x\in\xc: x \neq e(y)}P_{X|Y}(x|y)
\end{equation}
holds almost surely in $P_Y$, where
$e(y)$ is the MAP estimate from $y$ as defined in (\ref{eq:pv1}); 
in other words, the MAP estimate is almost surely unique in $P_Y$.
Then, 
the error probability in the MAP estimation of $X$ from $Y$ satisfies
\begin{equation}\label{exact-pe}
P_e=\lim_{\theta\rightarrow\infty}P_{X,Y}\left\{(x,y)\in\xc\times\yc:~
P_{X|Y}^{(\theta)}(x|y)\leq \alpha\right\}
\end{equation}
for each $\alpha\in(0,1)$, where the tilted distribution
$P_{X|Y}^{(\theta)}(\cdot|y)$  
is given in (\ref{twistdist}) for $y\in \yc$.
}\end{theorem}
  
\medskip
\begin{proof}
It can be easily verified from the definitions of $h_j(\cdot)$
and $h_j^{(\theta)}(\cdot)$ that the following two limits hold for each $y\in \yc$:
$$\lim_{\theta\rightarrow\infty}h_1^{(\theta)}(y)=\frac 1{\ell(y)},$$
where 
\begin{equation}
\ell(y)\triangleq\max\{j\in\mathbb{N}:h_j(y)=h_1(y)\}
\label{ell-def}
\end{equation}
and
$\mathbb{N}\triangleq\{1,2,3,\ldots\}$ is the set of positive integers,
and
\begin{equation}
\lim_{\theta \rightarrow \infty} 
h_j(y) \cdot\oneb\left(h^{(\theta)}_j(y)>\alpha\right)=
\begin{cases}
h_j(y) \cdot \oneb\left(\frac 1{\ell(y)}>\alpha\right) 
& \text{for $j=1,2,\cdots,{\ell(y)}$} \\
0 & \text{for $j > {\ell(y)}$}
\end{cases}
\label{limit-h}
\end{equation}
where $\oneb(\cdot)$ is the indicator function.

As a result, we obtain
that for any $\alpha \in [0,1)$,
\begin{eqnarray}
\lefteqn{\lim_{\theta\rightarrow\infty}P_{X,Y}\left\{(x,y)\in\xc\times\yc:~P_{X|Y}^{(\theta)}(x|y)>\alpha\right\}} \nonumber\\
&=&\lim_{\theta\rightarrow\infty}\int_\yc\left(\sum_{j=1}^\infty h_j(y)\cdot\oneb\left(h^{(\theta)}_j(y)>\alpha\right)\right)dP_Y(y)\nonumber\\
&=&\int_\yc\lim_{\theta\rightarrow\infty}\left(\sum_{j=1}^\infty h_j(y)\cdot\oneb\left(h^{(\theta)}_j(y)>\alpha\right)\right)dP_Y(y)\label{dom}\\
&=&\int_\yc\left(\sum_{j=1}^{\ell(y)} h_j(y)
\cdot \oneb\left(\frac 1{\ell(y)}>\alpha\right)\right)dP_Y(y),\label{limit-sum}
\end{eqnarray}
where \eqref{dom} follows from the Dominated Convergence Theorem \cite[Thm.~16.4]{billingsley} 
since
$$\left|\sum_{j=1}^\infty h_j(y)\cdot\oneb\left(h^{(\theta)}_j(y)>\alpha\right)\right|
\leq \sum_{j=1}^\infty h_j(y)=1.$$
Furthermore, (\ref{limit-sum}) holds since the limit (in $\theta$) of 
$$a_{\theta,j} \triangleq h_j(y) \cdot \oneb\left(h^{(\theta)}_j(y)>\alpha\right)$$
exists for every $j=1,2,\cdots$ by (\ref{limit-h}), hence implying 
(as shown in Appendix~A) that
$$\lim_{\theta \rightarrow \infty} \sum_{j=1}^{\infty} a_{\theta,j}
= \sum_{j=1}^{\infty} \lim_{\theta \rightarrow \infty} a_{\theta,j}.$$

\noindent
Now condition (\ref{condition}) is equivalent to 
\begin{equation}
\Pr[\ell(Y)=1] \triangleq P_Y\left\{y\in \yc: 
\ell(y)=1\right\}=1;
\label{ell-cond}
\end{equation}
thus,
\begin{eqnarray}
\lim_{\theta\rightarrow\infty}P_{X,Y}\left\{(x,y)\in
\xc\times\yc:~P_{X|Y}^{(\theta)}(x|y)>\alpha\right\}
&=&\int_\yc h_1(y)\cdot\oneb\left(1>\alpha\right)dP_Y(y) 
=E[h_1(Y)] \nonumber \\
&=&1-P_e, \label{h1-Pe}
\end{eqnarray}
where (\ref{h1-Pe}) follows from (\ref{eq:firstpv}).

This immediately yields that for $0<\alpha<1$,
\begin{eqnarray*}
P_e&=&1-\lim_{\theta\rightarrow\infty}P_{X,Y}\left\{(x,y)\in\xc\times\yc:~P_{X|Y}^{(\theta)}(x|y)>\alpha\right\} \\
&=&\lim_{\theta\rightarrow\infty}P_{X,Y}\left\{(x,y)\in\xc\times\yc:~P_{X|Y}^{(\theta)}(x|y)\leq\alpha\right\}.
\end{eqnarray*}
\end{proof}

\bigskip

\begin{obs}
{\rm 
We first note that since the bound in (\ref{gen-pv-bound}) holds for every 
$\theta \ge 1$, it also holds in the limit of $\theta$ going to infinity (the limit exists as shown in the above proof):
\begin{equation}
P_e \ge (1-\alpha)\lim_{\theta\rightarrow\infty}P_{X,Y}\left\{
(x,y)\in\xc\times\yc:~P_{X|Y}^{(\theta)}(x|y)\leq \alpha\right\}
\label{limit-gen-pv-bound}
\end{equation}
for any $0 \leq \alpha \leq 1$.

Furthermore, 
if condition \eqref{condition} does not hold (or equivalently from (\ref{ell-cond}), 
if $\Pr[\ell(Y)=1]<1$), but
there exists an integer $L>1$ such that $\Pr[\ell(Y)\leq L]=1$, then 
using (\ref{limit-sum}), we can write (\ref{limit-gen-pv-bound}) as
\begin{eqnarray}
P_e&\geq&(1-\alpha) \left[1-\int_\yc\left(\sum_{j=1}^{\ell(y)} h_j(y)\cdot\oneb
\left(\frac 1{\ell(y)}>\alpha\right)\right)dP_Y(y)\right] \nonumber \\
&=&(1-\alpha) \left[\int_\yc \left(\sum_{j=1}^{\infty} h_j(y)\right)dP_Y(y)-\int_\yc\left(\sum_{j=1}^{\ell(y)} h_j(y)\cdot\oneb
\left(\frac 1{\ell(y)}>\alpha\right)\right)dP_Y(y)\right] \nonumber \\
&=&(1-\alpha)\int_\yc\left(\sum_{j=1}^{\ell(y)} h_j(y)\cdot\oneb\left(\frac 1{\ell(y)}\leq\alpha\right)+\sum_{j=\ell(y)+1}^{\infty} h_j(y)\right)dP_Y(y) \label{a1}\\
&=&(1-\alpha)\left[\int_{y:\ell(y)=1}\left(\sum_{j=1}^{1} h_j(y)\cdot\oneb\left(1\leq\alpha\right)
+\sum_{j=2}^{\infty} h_j(y)\right)dP_Y(y)\right.\nonumber\\
&&+\int_{y:\ell(y)=2}\left(\sum_{j=1}^{2} h_j(y)\cdot\oneb\left(\frac 1{2}\leq\alpha\right)+\sum_{j=3}^{\infty} h_j(y)\right)dP_Y(y)\nonumber\\
&&\left.+\cdots+\int_{y:\ell(y)=L}\left(\sum_{j=1}^{L} h_j(y)\cdot\oneb\left(\frac 1{L}\leq\alpha\right)+\sum_{j=L+1}^{\infty} h_j(y)\right)dP_Y(y)\right].\label{a2}
\end{eqnarray}
To render this lower bound 
as large as possible, its formula above indicates that although 
the multiplicative constant $(1-\alpha)$ favors a small $\alpha$,
the integration term in \eqref{a1} actually has its smallest value 
when $\alpha$ is less than $1/L$ (see (\ref{a2})).
Therefore, a compromise in the choice of $\alpha$ has to be made in order 
to maximize the bound.
}
\end{obs}

\section{Examples for the generalized Poor-Verd\'{u} bound}\label{examples}

In this section, we provide four examples (three of them with a finite observation 
alphabet and one with a continuous observation alphabet) to illustrate
the results of the previous section.

\subsection{Ternary Hypothesis Testing}\label{pv-example}
We revisit the ternary hypothesis testing example examined
in \cite[Figs.~1 and~2]{PV95}, where random variables $X$ and $Y$
have identical alphabets $\xc=\yc=\{0,1,2\}$, $X$ is uniformly distributed 
($P_X(x)=1/3 \ \forall x\in \xc$) and $Y$ is related to $X$ 
via
$$P_{Y|X}(y|x)=\begin{cases}
1-v_1-v_2&{\rm if} \ y=x\\
v_1&{\rm if} \ x=1\ {\rm and}\ y=0\\
v_2&{\rm if} \ x=2\ {\rm and}\ y=0\\
v_1&{\rm if} \ y\neq x\ {\rm and}\ y=1\\
v_2&{\rm if} \ y\neq x\ {\rm and}\ y=2
\end{cases}
$$
where we assume that $1-v_1-v_2>v_2>v_1>0$.
In \cite{PV95}, $v_1=0.27$ and $v_2=0.33$ are used.

A direct calculation reveals that the MAP estimation function
(\ref{eq:pv1}) for guessing $X$ from $Y$ is given by $e(y)=y$
for every $y\in \yc$, resulting in a probability of error
of $P_e=v_1+v_2=0.6$ when $v_1=0.27$ and $v_2=0.33$. 
Furthermore, we obtain that $P_e$ is exactly determined via
$$
\lim_{\theta\rightarrow\infty}P_{X,Y}\left\{(x,y)\in\xc\times\yc:~
P_{X|Y}^{(\theta)}(x|y)\leq \alpha\right\}= v_1+v_2 = P_e;$$
as predicted by Theorem~\ref{tightness-theo}, since
condition~(\ref{condition}) holds (since $\ell(Y)=1$ almost surely in $P_Y$,
where $\ell(\cdot)$ is defined in \eqref{ell-def}).

We next compute the new bound in (\ref{gen-pv-bound}) 
for $v_1=0.27$, $v_2=0.33$ and for different values
of $\theta \ge 1$ and plot it in Fig.~\ref{pv-examp-plot1}, along with
Fano's original bound (referred to as ``Fano'' in the figure) given by
$$P_e \ge \frac{\log 3-I(X;Y)-\log 2}{\log2}=0.568348,$$
and Fano's weaker (but commonly used) bound 
$$P_e \ge 1-\frac{I(X;Y)+\log 2}{\log 3}=0.358587$$
shown in \cite[Fig.~2]{PV95} (referred to as ``Weakened Fano'' in the figure).
The case of $\theta=1$ corresponds to the
original Poor-Verd\'{u} bound in (\ref{pv-bound}).
As can be seen from the figure, bound (\ref{gen-pv-bound})
for $\theta=20$ and 100 improves upon (\ref{pv-bound}) and both Fano bounds 
and approaches the exact probability of error 
as $\theta$ is increased without bound (e.g., for $\theta=100$ and $\alpha\downarrow 0$, 
the bound is quite close to $P_e$). 
In Fig.~\ref{pv-examp-plot2}, bounds
(\ref{gen-pv-bound}) and (\ref{pv-bound}), maximized over $\alpha \in [0,1]$,
are plotted versus $\theta$. It is observed that when $\theta \ge 16$,
bound (\ref{gen-pv-bound}) improves upon (\ref{pv-bound}). 

\begin{figure}[h!]
\begin{center}
\input{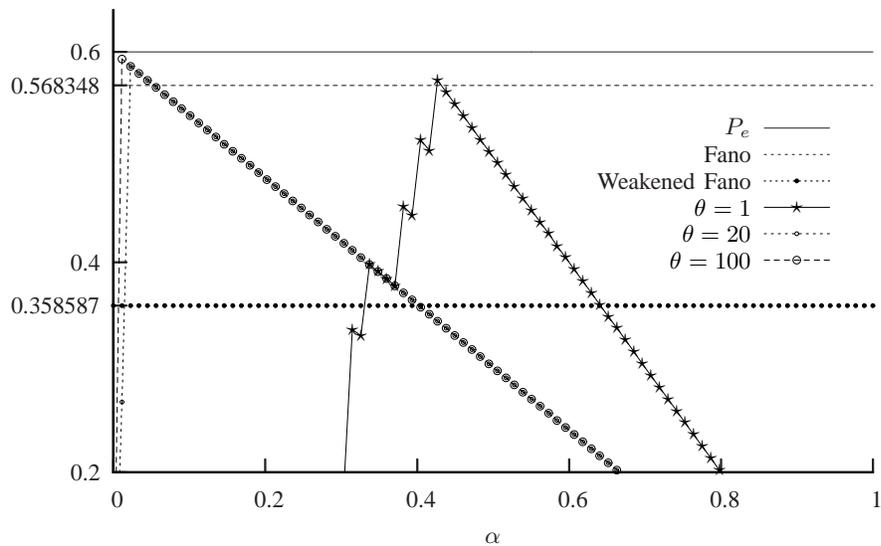}
\end{center}
\vspace{-0.2in}
\caption{Lower bounds on the minimum probability of error for Example~\ref{pv-example}: 
bound (\ref{gen-pv-bound}) versus $\alpha$ for $\theta=1,20,100$
and Fano's original and weakened bounds.}
\label{pv-examp-plot1}
\end{figure}

\begin{figure}[h!]
\begin{center}
\setlength{\unitlength}{0.120450pt}
\begin{picture}(3000,1800)(0,0)
\footnotesize
\thicklines \path(512,377)(553,377)
\put(471,377){\makebox(0,0)[r]{0.4}}
\thicklines \path(512,936)(553,936)
\put(471,936){\makebox(0,0)[r]{0.5}}
\thicklines \path(512,1352)(553,1352)
\put(471,1352){\makebox(0,0)[r]{0.574468}}
\thicklines \path(512,1494)(553,1494)
\put(471,1494){\makebox(0,0)[r]{0.6}}
\thicklines \path(638,265)(638,306)
\put(638,182){\makebox(0,0){ 2}}
\thicklines \path(889,265)(889,306)
\put(889,182){\makebox(0,0){ 4}}
\thicklines \path(1140,265)(1140,306)
\put(1140,182){\makebox(0,0){ 6}}
\thicklines \path(1391,265)(1391,306)
\put(1391,182){\makebox(0,0){ 8}}
\thicklines \path(1642,265)(1642,306)
\put(1642,182){\makebox(0,0){ 10}}
\thicklines \path(1893,265)(1893,306)
\put(1893,182){\makebox(0,0){ 12}}
\thicklines \path(2144,265)(2144,306)
\put(2144,182){\makebox(0,0){ 14}}
\thicklines \path(2395,265)(2395,306)
\put(2395,182){\makebox(0,0){ 16}}
\thicklines \path(2646,265)(2646,306)
\put(2646,182){\makebox(0,0){ 18}}
\thicklines \path(2897,265)(2897,306)
\put(2897,182){\makebox(0,0){ 20}}
\thicklines \path(512,1718)(512,265)(2897,265)
\put(1704,58){\makebox(0,0){$\theta$}}
\put(2108,783){\makebox(0,0)[r]{$P_e$}}
\thinlines \path(2149,783)(2354,783)
\thinlines \path(512,1494)(512,1494)(573,1494)(634,1494)(695,1494)(757,1494)(818,1494)(879,1494)(940,1494)(1001,1494)(1062,1494)(1124,1494)(1185,1494)(1246,1494)(1307,1494)(1368,1494)(1429,1494)(1490,1494)(1552,1494)(1613,1494)(1674,1494)(1735,1494)(1796,1494)(1857,1494)(1919,1494)(1980,1494)(2041,1494)(2102,1494)(2163,1494)(2224,1494)(2285,1494)(2347,1494)(2408,1494)(2469,1494)(2530,1494)(2591,1494)(2652,1494)(2714,1494)(2775,1494)(2836,1494)(2897,1494)
\put(2108,700){\makebox(0,0)[r]{maximized bound for $\theta=1$}}
\thinlines \dashline[90]{10}(2149,700)(2354,700)
\thinlines \dashline[90]{10}(512,1352)(512,1352)(573,1352)(634,1352)(695,1352)(757,1352)(818,1352)(879,1352)(940,1352)(1001,1352)(1062,1352)(1124,1352)(1185,1352)(1246,1352)(1307,1352)(1368,1352)(1429,1352)(1490,1352)(1552,1352)(1613,1352)(1674,1352)(1735,1352)(1796,1352)(1857,1352)(1919,1352)(1980,1352)(2041,1352)(2102,1352)(2163,1352)(2224,1352)(2285,1352)(2347,1352)(2408,1352)(2469,1352)(2530,1352)(2591,1352)(2652,1352)(2714,1352)(2775,1352)(2836,1352)(2897,1352)
\put(2108,617){\makebox(0,0)[r]{maximized bound for $\theta \ge 1$}}
\thinlines \dottedline{10}(2149,617)(2354,617)
\thinlines \dottedline{10}(512,1352)(512,1352)(573,1087)(634,820)(695,555)(757,484)(818,519)(879,558)(940,600)(1001,643)(1062,688)(1124,733)(1185,778)(1246,822)(1307,865)(1368,907)(1429,948)(1490,986)(1552,1023)(1613,1058)(1674,1091)(1735,1123)(1796,1152)(1857,1179)(1919,1204)(1980,1228)(2041,1250)(2102,1270)(2163,1289)(2224,1306)(2285,1322)(2347,1336)(2408,1350)(2469,1362)(2530,1374)(2591,1384)(2652,1394)(2714,1402)(2775,1410)(2836,1418)(2897,1424)
\put(512,1352){\makebox(0,0){$\star$}}
\put(573,1087){\makebox(0,0){$\star$}}
\put(634,820){\makebox(0,0){$\star$}}
\put(695,555){\makebox(0,0){$\star$}}
\put(757,484){\makebox(0,0){$\star$}}
\put(818,519){\makebox(0,0){$\star$}}
\put(879,558){\makebox(0,0){$\star$}}
\put(940,600){\makebox(0,0){$\star$}}
\put(1001,643){\makebox(0,0){$\star$}}
\put(1062,688){\makebox(0,0){$\star$}}
\put(1124,733){\makebox(0,0){$\star$}}
\put(1185,778){\makebox(0,0){$\star$}}
\put(1246,822){\makebox(0,0){$\star$}}
\put(1307,865){\makebox(0,0){$\star$}}
\put(1368,907){\makebox(0,0){$\star$}}
\put(1429,948){\makebox(0,0){$\star$}}
\put(1490,986){\makebox(0,0){$\star$}}
\put(1552,1023){\makebox(0,0){$\star$}}
\put(1613,1058){\makebox(0,0){$\star$}}
\put(1674,1091){\makebox(0,0){$\star$}}
\put(1735,1123){\makebox(0,0){$\star$}}
\put(1796,1152){\makebox(0,0){$\star$}}
\put(1857,1179){\makebox(0,0){$\star$}}
\put(1919,1204){\makebox(0,0){$\star$}}
\put(1980,1228){\makebox(0,0){$\star$}}
\put(2041,1250){\makebox(0,0){$\star$}}
\put(2102,1270){\makebox(0,0){$\star$}}
\put(2163,1289){\makebox(0,0){$\star$}}
\put(2224,1306){\makebox(0,0){$\star$}}
\put(2285,1322){\makebox(0,0){$\star$}}
\put(2347,1336){\makebox(0,0){$\star$}}
\put(2408,1350){\makebox(0,0){$\star$}}
\put(2469,1362){\makebox(0,0){$\star$}}
\put(2530,1374){\makebox(0,0){$\star$}}
\put(2591,1384){\makebox(0,0){$\star$}}
\put(2652,1394){\makebox(0,0){$\star$}}
\put(2714,1402){\makebox(0,0){$\star$}}
\put(2775,1410){\makebox(0,0){$\star$}}
\put(2836,1418){\makebox(0,0){$\star$}}
\put(2897,1424){\makebox(0,0){$\star$}}
\put(2251,617){\makebox(0,0){$\star$}}
\thicklines \path(512,1718)(512,265)(2897,265)
\end{picture}
\end{center}
\vspace{-0.2in}
\caption{Lower bounds on the minimum probability of error for Example~\ref{pv-example}:
bounds (\ref{pv-bound}) and (\ref{gen-pv-bound}) versus $\theta$ optimized over $\alpha$.}
\label{pv-examp-plot2}
\end{figure}
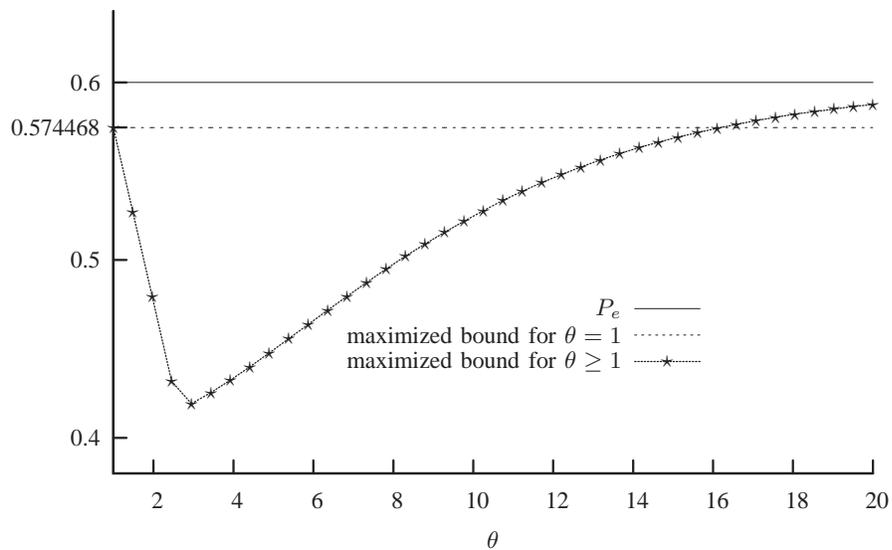

\subsection{Binary Erasure Channel}

Suppose that $X$ and $Y$ are respectively the channel input and output of a 
BEC with erasure probability $\ep$, where $\xc=\{0,1\}$ 
and $\yc=\{0,1,{\tt E}\}$. Let $\Pr[X=0]=1-p$ and $\Pr[X=1]=p$ with $0<p<1/2$.
Then, the MAP estimate of $X$ from $Y$ is given by 
$$e(y)= \begin{cases}
y & \text{if $y\in \{0,1\}$} \\
0 & \text{if $y={\tt E}$}
\end{cases}
$$
and the resulting error probability is $P_e=\ep p$.

Calculating bound (\ref{gen-pv-bound}) of Theorem \ref{improved-pv}
yields
\begin{multline}\label{BECeq}
(1-\alpha)P_{X,Y}\left\{(x,y)\in\xc\times\yc:~
P_{X|Y}^{(\theta)}(x|y)\leq \alpha\right\}\\=
\begin{cases}
0& {\rm if} \ 0\leq\alpha<\displaystyle\frac{p^\theta}{p^\theta+(1-p)^\theta}\\
\ep p(1-\alpha)& {\rm if} \ \displaystyle\frac{p^\theta}{p^\theta+(1-p)^\theta}\leq\alpha<\frac{(1-p)^\theta}{p^\theta+(1-p)^\theta}\\
\ep(1-\alpha)& {\rm if} \ \displaystyle\frac{(1-p)^\theta}{p^\theta+(1-p)^\theta}\leq\alpha<1.
\end{cases}
\end{multline}
Thus, taking $\theta\uparrow\infty$ and then $\alpha\downarrow 0$ in 
(\ref{BECeq}) results in the exact error probability $\ep p$. 
Note that in this example, the original Poor-Verd\'{u} bound (i.e., with $\theta=1$) 
also achieves the exact error probability $\ep p$ by choosing $\alpha=1-p$;
however this maximizing choice of $\alpha=1-p$ for the original bound
is a function of system's statistics
(here, the input distribution $p$) which is undesirable.
On the other hand, the generalized bound (\ref{gen-pv-bound}) can herein achieve 
its peak by systematically taking $\theta\uparrow\infty$ and then letting 
$\alpha\downarrow 0$. 

Furthermore, since in this example, $\ell(y)=1$ for every $y\in\{0,1,{\tt E}\}$,
we have that (\ref{condition}) holds; hence, by Theorem~\ref{tightness-theo}, 
\eqref{exact-pe} yields
\begin{eqnarray*}
P_e
&=&\lim_{\theta\rightarrow\infty}P_{X,Y}\left\{(x,y)\in\xc\times\yc:~P_{X|Y}^{(\theta)}(x|y)\leq\alpha\right\}\\
&=&\ep p\mbox{ for }0\leq\alpha<1,
\end{eqnarray*}
where the last equality follows directly from \eqref{BECeq} without the
$(1-\alpha)$ factor.

\subsection{Multiple-Use BEC}

We now extend the previous example of the single-use BEC to the case
of using the memoryless BEC $n$ times with an input $n$-tuple $X^n=(X_1,\cdots,X_n)$
of independent and identically distributed
(i.i.d.) random variables $X_i$ with $\Pr[X_i=1]=p$, where $0<p<1/2$. 
Here again we determine the MAP estimation of $X^n$ by observing the channel output $Y^n$.
For a received output $n$-tuple $y^n$,
\begin{equation}
P_{X^n|Y^n}(x^n|y^n)= \begin{cases}
(1-p)^{d_{0{\tt E}}(x^n,y^n)}p^{d_{1{\tt E}}(x^n,y^n)}& \text{if $d_{01}(x^n,y^n)=d_{10}(x^n,y^n)=0$} \\
0&\text{otherwise}
\end{cases}
\label{8b}
\end{equation}
where $d_{0{\tt E}}(x^n,y^n)$ is the number of occurrences of $(x_j,y_j)=(0,{\tt E})$ 
in $(x^n,y^n)$, and the other $d$-terms are defined similarly.
The above equation indicates that for a given $y^n$, $P_{X^n|Y^n}(x^n|y^n)$ always peaks
for $d_{1{\tt E}}(x^n,y^n)=0$ since $0<p<1/2$. 
Thus the MAP estimator $e(y^n)$ replaces all erasures in $y^n$ by 0's
while keeping the 0's and 1's in $y^n$ unchanged (e.g., if $n=5$ and $y^n=(0,0,{\tt E},{\tt E},1)$,
then $e(y^n)=(0,0,0,0,1)$). The resulting probability of error is given by
\begin{eqnarray*}
P_e&=&1-\sum_{y^n\in\yc^n}P_{X^n}(e(y^n))P_{Y^n|X^n}(y^n|e(y^n))\\
&=&1-\sum_{k=0}^n\sum_{i=0}^{n-k}\binom nk\binom{n-k}i(1-p)^{n-i}p^i\ep^k(1-\ep)^{n-k}\\
&=&1-(1-\ep p)^n
\end{eqnarray*}
where $k$ is the number of erasures ${\tt E}$ in $y^n$
and $i$ is the number of 1's in $y^n$.

On the other hand, we directly obtain from \eqref{8b}
that condition (\ref{condition}) holds (or equivalently
condition (\ref{ell-cond}), i.e., $\ell(y^n)=1$ with probability one in $P_{Y^n}$). 
We can then apply Theorem~\ref{tightness-theo} to obtain from \eqref{exact-pe} that
\begin{eqnarray*}
P_e&=& 1-(1-\ep p)^n \\
&=&\lim_{\theta\rightarrow\infty}P_{X^n,Y^n}\left\{(x^n,y^n)\in\xc\times\yc:~P_{X^n|Y^n}^{(\theta)}(x^n|y^n)\leq\alpha\right\}.
\end{eqnarray*}

We next consider the case of $p=1/2$, i.e,. the input $X^n$ is uniformly distributed. 
In this case, \eqref{8b} yields that
$$h_1(y^n)=h_2(y^n)=\cdots=h_{2^k}(y^n)=2^{-k}$$
and
$$h_{2^k+1}(y^n)=h_{2^k+2}(y^n)=\cdots=h_{2^n}(y^n)=0$$
where $k$ is the number of erasures ${\tt E}$ in $y^n$.
Thus $\ell(y^n)=2^k$ and Theorem~\ref{tightness-theo} no longer holds.
Furthermore, $h_j^{(\theta)}(y^n)=h_j(y^n)$ for every $\theta\geq 1$; this
implies that for the uniform-input multiple-use BEC, 
the generalized bound (\ref{gen-pv-bound}) does not improve upon
the original Poor-Verd\'{u} bound (\ref{pv-bound}).

\subsection{Binary Input Observed in Gaussian Noise}\label{binary-input-awgn}

We herein consider an example with a continuous observation alphabet
$\yc =\mathbb{R}$, where $\mathbb{R}$ is the set of real numbers. Specifically, 
let the observation be given by $Y=X+N$, where $X$ is uniformly distributed over 
$\xc=\{-1,+1\}$ and $N$ is a zero-mean Gaussian random variable with variance $\sigma^2$.
Assuming that $X$ and $N$ are independent from each other, then
\begin{eqnarray}
P_{X|Y}(x|y)&=&\frac{\frac 12\cdot\frac 1{\sqrt{2\pi\sigma^2}}\exp\{-\frac{(y-x)^2}
{2\sigma^2}\}}{\frac 12\cdot\frac 1{\sqrt{2\pi\sigma^2}}\exp\{-\frac{(y-1)^2}
{2\sigma^2}\}
+\frac 12\cdot\frac 1{\sqrt{2\pi\sigma^2}}\exp\{-\frac{(y+1)^2}{2\sigma^2}\}} \nonumber\\
&=&\frac{\exp\{\frac{xy}{\sigma^2}\}}{\exp\{\frac{y}{\sigma^2}\}
+\exp\{-\frac{y}{\sigma^2}\}} 
=\frac{1}{1 +\exp\{-\frac{2xy}{\sigma^2}\}} \label{cond-x-y}
\end{eqnarray}
for $x\in \{-1,+1\}$, $y \in \mathbb{R}$. This directly implies that the MAP estimate
of $X$ from $Y$ is given by $e(y)=+1$ if $y>0$ and $e(y)=-1$ if $y\le0$.
The resulting error probability is $P_e = \Phi(-1/\sigma)$,
where $\Phi(z) \triangleq \frac{1}{\sqrt{2\pi}} \int_{-\infty}^{z}\exp{-\frac{t^2}{2}} dt$ 
is the cdf of the standard (zero-mean unit-variance) 
Gaussian distribution.
Furthermore, since $x\in\{-1,+1\}$, we have
$$P_{X|Y}^{(\theta)}(x|y)
=\frac{\left(\frac{\exp\{\frac{xy}{\sigma^2}\}}{\exp\{\frac{y}{\sigma^2}\}
+\exp\{-\frac{y}{\sigma^2}\}}\right)^\theta}{\left(
\frac{\exp\{\frac{y}{\sigma^2}\}}{\exp\{\frac{y}{\sigma^2}\}
+\exp\{-\frac{y}{\sigma^2}\}}\right)^\theta+\left(\frac{\exp\{\frac{-y}{\sigma^2}\}}
{\exp\{\frac{y}{\sigma^2}\}
 +\exp\{-\frac{y}{\sigma^2}\}}\right)^\theta}
=\frac{1}{1 +\exp\{-\frac{2xy}{\sigma^2/\theta}\}}, 
$$
and the generalized Poor-Verd\'{u} bound (\ref{gen-pv-bound}) yields
\begin{eqnarray}
P_e&\geq&(1-\alpha)P_{X,Y}\left\{(x,y)\in\xc\times\yc~:~
P_{X|Y}^{(\theta)}(x|y)\leq \alpha\right\}\nonumber\\
&=&(1-\alpha)P_{X}(-1)
\int_{y\in\mathbb{R}~:~
\frac{1}{1 +\exp\left\{\frac{2y}{\sigma^2/\theta}\right\}}\leq \alpha}
\frac 1{\sqrt{2\pi\sigma^2}}\exp\left\{-\frac{(y+1)^2}{2\sigma^2}\right\}dy\nonumber\\
&&+(1-\alpha)P_{X}(1)\int_{y\in\mathbb{R}~:~ \frac{1}{1
 +\exp\left\{-\frac{2y}{\sigma^2/\theta}]\right\}}\leq \alpha}\frac 1{\sqrt{2\pi
\sigma^2}}\exp\left\{-\frac{(y-1)^2}{2\sigma^2}\right\}dy\nonumber\\
&=&\frac{(1-\alpha)}2\int_{\frac{\sigma^2}{2\theta}\log\left(\frac 1\alpha-1\right)}^\infty
\frac 1{\sqrt{2\pi\sigma^2}}\exp\left\{-\frac{(y+1)^2}{2\sigma^2}\right\}dy\nonumber\\
 &&+\frac{(1-\alpha)}2\int_{-\infty}
^{-\frac{\sigma^2}{2\theta}\log\left(\frac 1\alpha-1\right)}\frac 1{\sqrt{2\pi
\sigma^2}}\exp\left\{-\frac{(y-1)^2}{2\sigma^2}\right\}dy\nonumber\\
 &=&(1-\alpha)\int_{-\infty}^{-\frac{\sigma^2}{2\theta}\log\left(\frac 1\alpha-1
\right)-1}
\frac 1{\sqrt{2\pi\sigma^2}}\exp\left\{-\frac{t^2}{2\sigma^2}\right\}dt\nonumber\\
&=&(1-\alpha)\Phi\left(-\frac{\sigma}{2\theta}\log\left(\frac 1\alpha-1\right)-
\frac 1\sigma\right). \label{rhs-theta}
\end{eqnarray}

\bigskip

Now taking the limits $\theta\uparrow\infty$ followed by $\alpha\downarrow 0$
for the right-hand side term in \eqref{rhs-theta} yields exactly 
$\Phi\left(-\frac 1\sigma\right)=P_e$; hence the generalized Poor-Verd\'{u} bound 
(\ref{gen-pv-bound}) is asymptotically tight. The bound is illustrated in Fig.~\ref{awgn-examp}
for $\sigma=0.429858$ which gives $P_e=0.01$. It can be seen that for $\theta=100$
and $\alpha \downarrow 0$, bound (\ref{gen-pv-bound}) is quite close to $P_e$. 
Finally note that \eqref{cond-x-y} directly ascertains that condition \eqref{condition}
of Theorem~\ref{tightness-theo} holds; thus $P_e$ is given by \eqref{exact-pe}.
\begin{figure}[h!]
\begin{center}
\input{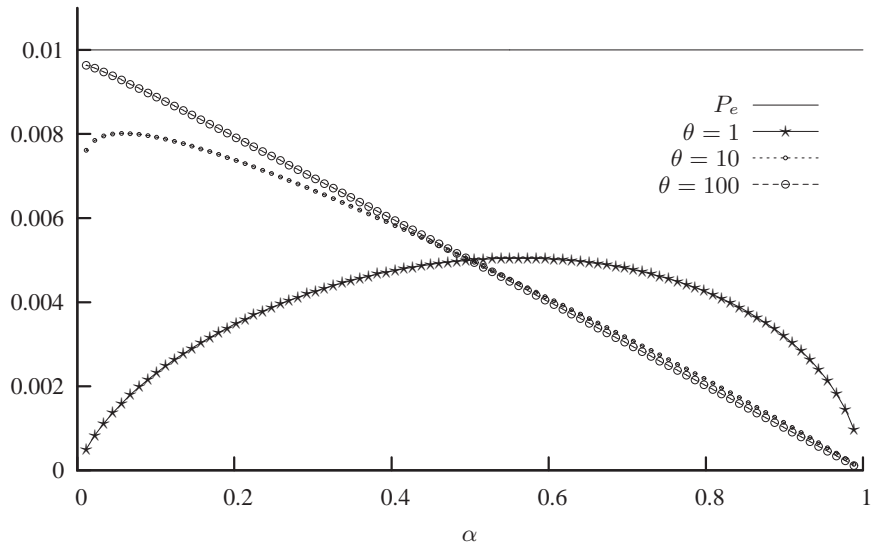}
\end{center}
\vspace{-0.2in}
\caption{Example~\ref{binary-input-awgn}: 
bound (\ref{gen-pv-bound}) versus $\alpha$ for $\theta=1,10,100$;
$\sigma=0.429858$ and $P_e=0.01$.}
\label{awgn-examp}
\end{figure}

\section{Channel reliability function}\label{channel-reliability}

We next use the results of Section~\ref{gen-pv-bound-dec} to study the channel
reliability function.

\subsection{Preliminaries}\label{prelim}

Consider an arbitrary input process {\bf X}  defined by a sequence of
finite-dimensional distributions \cite{VH94,Hanbook} $${\bf
X}\triangleq
\left\{X^n=\left(X^{(n)}_1,\cdots,X^{(n)}_n\right)\right\}^{\infty}_{n=1}
.$$
Denote by $${\bf
Y}\triangleq
\left\{Y^n=\left(Y^{(n)}_1,\cdots,Y^{(n)}_n\right)\right\}^{\infty}_{n=1}
$$ the corresponding output process induced by {\bf X} via a general channel with memory
$${\bf W}
\triangleq\{
W^n=P_{Y^n|X^n}: {\cal X}^n \rightarrow {\cal Y}^n\}_{n=1}^{\infty}$$
which is an arbitrary sequence of $n$-dimensional conditional distributions
from ${\cal X}^n$ to ${\cal Y}^n$, where $\cal X$ and $\cal Y$ are the input
and output alphabets, respectively. 

We assume throughout this section
that $\cal X$ is finite and that $\cal Y$ is arbitrary.
Note though that for the sake of clarity, we adopt the notations of a
discrete probability space for $\yc$ with the usual caveats (such as replacing summations with 
integrals and working with the appropriate probability measures, e.g., 
see \cite[Remark~3.2.1]{Hanbook}). 
 
\bigskip

\begin{definition}[Channel block code]
{\rm An $(n,M)$ code $\code_n$ for channel $\Wb$ with input alphabet $\cal X$ and
output alphabet $\cal Y$ is a pair of maps $(f,g)$,
where 
$$f: \{1,2,\cdots,M\}\rightarrow {\cal X}^n$$
is the encoding function yielding codewords $f(1), f(2), \cdots, f(M) \in \xc^n$,
each of length $n$, and $$g: {\cal Y}^n \rightarrow \{1,2,\cdots,M\}$$
is the decoding function.
The set of the $M$ codewords
is called the codebook and we also usually write
$\code_n = \{f(1), f(2), \cdots, f(M)\}$ to list the codewords.
}
\end{definition}

The set $\{1,2,\ldots,M\}$ is called the message set and we assume
that a message $V$ is drawn from the message set according 
to the uniform distribution. To convey message $V$ over channel $\Wb$, 
its corresponding codeword $X^n=f(V)$ is sent over the channel.
Then $Y^n$ is received at the channel output
and $\hat{V}=g(Y^n)$ is yielded as the message estimate.

The code's average error probability (or average probability of decoding error) is given by
$$P_e(\code_n) \triangleq
\frac 1 M \sum _{m=1}^M \sum_{\{y^n: g(y^n)
\neq m\}}W^n(y^n|f(m)).$$
Since message $V$ is uniformly distributed over $\{1,2,\ldots,M\}$, we have that
$P_e(\code_n)=\Pr[V \neq \hat{V}]$.

\bigskip

\begin{definition}[Channel reliability function \cite{PV95}]\label{pv-exp-def}
{\rm
For any $R  > 0$, define the channel reliability function $E^\ast(R)$
for a channel $\Wb$ as
the largest scalar $\beta>0$ such that
there exists a sequence of $\code_n=(n,M_n)$ codes with\footnote{
When no $\beta>0$ satisfies the definition, we simply set $E^\ast(R)=0$.
}
$$
\beta\leq\liminf_{n\rightarrow\infty}-\frac 1n\log P_e(\code_n)
$$
and
\begin{equation}
\label{eq:13}
R < \liminf_{n\rightarrow\infty}\frac 1n\log M_n.
\end{equation}
}
\end{definition}

\bigskip

\begin{obs}{\rm \
We have adopted the above definition of channel reliability function
from \cite{PV95} for the sake of consistency. Note that
this definition is not exactly identical to the traditional
definition of the channel reliability function.
If $P_{e,\text{min}}(n,R)$ denotes the probability of error of 
the best $(n,\lceil 2^{nR}\rceil)$ code (i.e., the code with
smallest error probability) for channel $\Wb$, then
the channel's reliability function is traditionally defined as
\footnote{The limit supremum is also commonly used instead of the limit infimum
in the definition of $E(R)$, e.g., see \cite[p.~160]{gallager}.
We could have also used the limit supremum in the inequality on
$\beta$ in Definition~\ref{pv-exp-def}; in that case the results of this section
would still hold by replacing $\liminf_n$ 
with $\limsup_n$ in Theorems~\ref{exp-upper-bound}
and~\ref{awgn-exact-exp-theo} and Corollary~\ref{exact-exp-cor}.}
$$E(R)=\liminf_{n\rightarrow\infty}-\frac 1n\log P_{e,\text{min}}(n,R).$$ 
However, the following relation can be shown between
$E^\ast(R)$ and $E(R)$:
$$E(R)\geq E^\ast(R)\geq \lim_{\delta\downarrow 0}E(R+\delta).$$
Thus the above two definitions are equivalent
except possibly for discontinuity rate points (of which there are at most
countably many as $E^\ast(R)$ and $E(R)$ are non-increasing in $R$).
}
\end{obs}

\bigskip

\begin{definition}[\cite{VH94}]
{\rm
Given that $Y^n$ is the output of channel $W^n=P_{Y^n|X^n}$
due to input $X^n$ with distribution $P_{X^n}$, the channel information density
is defined as
\begin{equation}
i_{X^n W^n}(x^n;y^n) \triangleq \log \frac {W^n(y^n|x^n)}{P_{Y^n}(y^n)}
= \log \frac{P_{Y^n|X^n}(y^n|x^n)}{\sum_{\hat x^n \in \xc^n} 
P_{X^n}(\hat x^n) P_{Y^n|X^n}(y^n|\hat x^n)}
\label{info-density}
\end{equation}
for $(x^n,y^n) \in \xc^n \times \yc^n$.
}
\end{definition}

\bigskip

\begin{definition}
{\rm 
Fix $R>0$. For an input $\bf X$ and a channel $\bf W$, 
\begin{equation}
\piund_\Xb(R)~\triangleq~
\liminf_{n \rightarrow \infty} -\frac 1n \log 
P_{X^nW^n}\left\{(x^n,y^n)\in\xc^n\times\yc^n: \frac 1ni_{X^nW^n}(x^n;y^n)\leq R\right\}
\label{large-dev-spectrum}
\end{equation}
is called a large-deviation rate function for the normalized information density 
$\frac{1}{n}i_{X^nW^n}(\cdot,\cdot)$.
}
\end{definition}

\bigskip

\begin{prop}[Poor-Verd\'{u} upper bound to $E^\ast(R)$]
{\rm 
For a given channel $\Wb$, 
its reliability function $E^\ast(R)$ satisfies
\cite[Eq.~(14)]{PV95}, \cite[Theorem~1]{pv02}
\begin{equation}
E^\ast(R) \leq \sup_{\bf X}{\pi}_{\bf X}(R)
\label{pv-exp} 
\end{equation}
for any $R>0$, where
$\piund_\Xb(R)$ is the large-deviation rate function for $\frac{1}{n}i_{X^nW^n}(\cdot,\cdot)$
as defined in \eqref{large-dev-spectrum}.

Furthermore, the bound in \eqref{pv-exp} can be slightly tightened
by restricting the supremum operation over a smaller set of inputs
\cite[Corollary~1]{pv02}:
\begin{equation}
E^\ast(R) \leq E_{\text{PV}}(R) \triangleq \sup_{\Xb\in\qc(R)}\piund_\Xb(R),
\label{tightened-pv-exp}
\end{equation}
for any $R > 0$,
where
\begin{multline}
\qc(R)\triangleq\bigg\{\Xb~:~
{\rm Each}\ X^n\ {\rm in}\ \Xb\ {\rm is\ uniformly\ distributed\ over\ its \
support}\
\scc(X^n),\\
{\rm and}\ R < \liminf_{n\rightarrow\infty}\frac
1n\log|\scc(X^n)|\bigg\}.
\label{input-set}
\end{multline}
}
\end{prop}

\subsection{Upper Bounds for the Channel Reliability Function}

Using Theorem~\ref{improved-pv}, we provide a lower bound for the probability of decoding error
of any $(n,M)$ channel code and establish two information-spectrum
upper bounds for the channel reliability
function.

\bigskip

\begin{theorem}\label{exp-upper-bound}{\rm Every $\code_n=(n,M)$ code for channel $\Wb$
has its probability of decoding error satisfying
\begin{equation}\label{codebound}
P_e(\code_n)\geq\left(1-\alpha\right)
P_{X^nW^n}\left\{(x^n,y^n)\in\xc^n\times\yc^n: j_{X^nW^n}^{(\theta)}(x^n;y^n) \leq \log(M\alpha) \right\}
\end{equation}
for every $\alpha\in[0,1]$ and $\theta\geq 1$, where channel input 
$X^n$ places probability mass $1/M$ on
each codeword of $\code_n$ and
\begin{equation}\label{ti}
j_{X^nW^n}^{(\theta)}(x^n;y^n)
\triangleq \log\displaystyle\frac{P_{Y^n|X^n}^\theta(y^n|x^n)}
{\sum_{\hat x^n\in\xc^n}P_{X^n}(\hat x^n)P_{Y^n|X^n}^\theta(y^n|\hat x^n)}.
\end{equation}
Furthermore, the channel's reliability function satisfies
\begin{eqnarray}\label{rebound-theta}
E^\ast(R)
&\leq& \sup_{\boldsymbol X\in{\cal Q}(R)}\liminf_{n\rightarrow\infty}-\frac 1n\log
P_{X^nW^n}\left\{(x^n,y^n)\in\xc^n\times\yc^n: \frac{1}{n}j_{X^nW^n}^{(\theta)}(x^n;y^n)
\leq R\right\} \nonumber\\
&\triangleq& E_{\text{PV}}^{(\theta)}(R)
\end{eqnarray}
for every $R>0$ and $\theta\geq 1$, and
\begin{eqnarray}\label{rebound}
E^\ast(R)
&\leq&\sup_{\boldsymbol X\in{\cal Q}(R)}\liminf_{n\rightarrow\infty}\lim_{\theta\rightarrow\infty}-\frac 1n\log
P_{X^nW^n}\left\{(x^n,y^n)\in\xc^n\times\yc^n: \frac{1}{n}j_{X^nW^n}^{(\theta)}(x^n;y^n)\leq R
\right\} \nonumber\\
&\triangleq& \bar{E}_{\text{PV}}(R)
\end{eqnarray}
for every $R>0$, where the set $\qc(R)$ is given in \eqref{input-set}.
}\end{theorem}

\bigskip
\begin{proof}
When the channel input $X^n$ is uniformly distributed over the code $\code_n \subseteq \xc^n$
of size $M$, the tilted distribution $P_{X^n|Y^n}^{(\theta)}$
of Theorem~\ref{improved-pv} becomes
\begin{eqnarray}
P_{X^n|Y^n}^{(\theta)}(x^n|y^n)&=&\frac{P_{X^n|Y^n}^\theta(x^n|y^n)}
{\sum_{\hat{x}^n\in\xc^n}P_{X^n|Y^n}^\theta(\hat{x}^n|y^n)} \nonumber\\
&=&\frac{P_{X^n}^\theta(x^n)P_{Y^n|X^n}^\theta(y^n|x^n)/P_{Y^n}^\theta(y^n)}
{\sum_{\hat{x}^n\in\xc^n}P_{X^n}^\theta(\hat{x}^n)P_{Y^n|X^n}^\theta(y^n|\hat{x}^n)/P_{Y^n}^\theta(y^n)} \nonumber\\
&=&\frac{P_{Y^n|X^n}^\theta(y^n|x^n)}
{\sum_{\hat{x}^n\in\xc^n}P_{Y^n|X^n}^\theta(y^n|\hat{x}^n)} \nonumber\\
&=&\frac{P_{Y^n|X^n}^\theta(y^n|x^n)/M}
{\sum_{\hat{x}^n\in\xc^n}P_{X^n}(\hat{x}^n)P_{Y^n|X^n}^\theta(y^n|\hat{x}^n)}
\label{uniform-inp}
\end{eqnarray}
for all $x^n \in \code_n$.
Hence inequality \eqref{codebound} follows directly from 
Theorem~\ref{improved-pv} and (\ref{uniform-inp}). We next prove (\ref{rebound});
the proof of (\ref{rebound-theta}) is identical by omitting the limit over $\theta$.
Setting $\alpha=e^{-n\gamma}$ in \eqref{codebound} yields
\begin{eqnarray*}
-\frac 1n\log P_e(\code_n)&\leq&-\frac 1n\log\left(1-e^{-n\gamma}\right) \\
&&-\frac 1n\log P_{X^nW^n}\left\{(x^n,y^n)\in \xc^n \times \yc^n:
\frac 1nj_{X^nW^n}^{(\theta)}(x^n;y^n)\leq \frac 1n\log M-\gamma\right\},
\end{eqnarray*}
which implies in light of \eqref{limit-gen-pv-bound}
\begin{eqnarray*}
\liminf_{n\rightarrow\infty}-\frac 1n\log P_e(\code_n)&\leq&
\liminf_{n\rightarrow\infty}\lim_{\theta\rightarrow\infty}
-\frac 1n\log P_{X^nW^n}\bigg\{(x^n;y^n)\in \xc^n\times\yc^n:  \\
&& \hspace{2.0in} \frac 1nj_{X^nW^n}^{(\theta)}(x^n;y^n)\leq \frac 1n\log M-\gamma\bigg\}.
\end{eqnarray*}
We can then conclude by definition of the channel reliability function that
\begin{eqnarray*}
E^\ast(R)&=&\sup_{\{\codeb_n=\scc(X^n)\}_{n\geq 1}:\Xb\in\qc(R)}
\liminf_{n\rightarrow\infty}-\frac 1n\log P_e(\code_n)\\
&\leq&\sup_{\Xb\in\qc(R)}\liminf_{n\rightarrow\infty}\lim_{\theta\rightarrow\infty}
-\frac 1n\log P_{X^nW^n}\bigg\{(x^n;y^n)\in \xc^n\times\yc^n:  \\
&& \hspace{2.0in} \frac 1nj_{X^nW^n}^{(\theta)}(x^n;y^n)\leq \frac 1n\log|\scc(X^n)|-\gamma\bigg\}.
\end{eqnarray*}
When considering only the sequence of codes in $\qc(R)$, we can replace $\frac 1n\log|\scc(X^n)|-\gamma$ by $R$ (if $\gamma$ is chosen to be small enough such that
$R<\liminf_{n\rightarrow\infty}\frac 1n\log|\scc(X^n)|-\gamma$ is valid for the 
considered input $\Xb$)
as such a replacement can only (ultimately) increase the upper bound; we thus obtain
$$E^\ast(R)\leq \sup_{\Xb\in\qc(R)}\liminf_{n\rightarrow\infty}\lim_{\theta\rightarrow\infty}
-\frac 1n\log P_{X^nW^n}\bigg\{(x^n;y^n)\in \xc^n\times\yc^n:  
\frac 1nj_{X^nW^n}^{(\theta)}(x^n;y^n)\leq R\bigg\}.$$
\end{proof}

\bigskip

\begin{obs}{\rm \
When $\theta=1$, $j_{X^nW^n}^{(\theta)}(x^n;y^n)$ in (\ref{ti}) reduces to
$$\log\displaystyle\frac{P_{Y^n|X^n}(y^n|x^n)}
{\sum_{\hat x^n\in\xc^n}P_{X^n}(\hat x^n)P_{Y^n|X^n}(y^n|\hat x^n)}
=\log\displaystyle\frac{P_{Y^n|X^n}(y^n|x^n)}{P_{Y^n}(y^n)}
= i_{X^nW^n}(x^n;y^n)$$
which is the channel information density as 
defined in \eqref{info-density}.

In this case, the generalized upper bound for the channel reliability function
$E_{\text{PV}}^{(\theta)}(R)$ of \eqref{rebound-theta} reduces to 
the Poor-Verd\'{u} upper bound $E_{\text{PV}}(R)$ of \eqref{tightened-pv-exp} 
(as expected, since for $\theta=1$, (\ref{gen-pv-bound}) reduces to (\ref{pv-bound})).

}
\end{obs}

\bigskip

\begin{obs}\label{auxiliary-channel}{\rm \
Note that when $\theta>1$,
the denominator of the fraction in \eqref{ti}
(in other words, $\sum_{\hat x^n\in\xc^n}P_{X^n}(\hat x^n)P_{Y^n|X^n}^\theta(y^n|\hat x^n))$
is not a legitimate distribution since it does not sum to one over $y^n\in\yc^n$.
However, if
\begin{equation}
\sum_{\hat y^n\in\yc^n}P_{Y^n|X^n}^\theta(\hat y^n|x^n)
= \sum_{\hat y^n\in\yc^n}P_{Y^n|X^n}^\theta(\hat y^n|\hat x^n)
\quad \quad \forall \ x^n, {\hat x^n} \in \xc^n, n=1,2,\cdots,
\label{invariance-cond}
\end{equation} 
then $j_{X^nW^n}^{(\theta)}(x^n;y^n)$ can be reformulated as follows
\begin{eqnarray} 
j_{X^nW^n}^{(\theta)}(x^n;y^n)&=&\log\displaystyle\frac{\displaystyle\frac{P_{Y^n|X^n}^\theta(y^n|x^n)}
{\sum_{\hat y^n\in\yc^n}P_{Y^n|X^n}^\theta (\hat y^n|x^n)}}
{\sum_{\hat x^n\in\xc^n}P_{X^n}(\hat x^n)
\displaystyle\frac{P_{Y^n|X^n}^\theta(y^n|\hat x^n)}{{\sum_{\hat y^n\in\yc^n}P_{Y^n|X^n}^\theta (\hat y^n|\hat x^n)}}}\nonumber\\
&=&\log\displaystyle\frac{P^{(\theta)}_{Y^n|X^n}(y^n|x^n)}
{\sum_{\hat x^n\in\xc^n}P_{X^n}(\hat x^n)
P^{(\theta)}_{Y^n|X^n}(y^n|x^n)}\label{nti} \\
&\triangleq& i_{X^n,Y^n}^{(\theta)}(x^n;y^n), \nonumber
\end{eqnarray}
where for each $y^n\in \yc^n$,
$$P^{(\theta)}_{Y^n|X^n}(y^n|x^n)\triangleq\displaystyle\frac{P_{Y^n|X^n}^\theta(y^n|x^n)}
{\sum_{\hat y^n\in\yc^n}P_{Y^n|X^n}^\theta (\hat y^n|x^n)} \quad x^n \in \xc^n$$
is the {\it tilted} distribution with parameter $\theta$ 
of the channel statistics $P_{Y^n|X^n}(\cdot|x^n)$. 
Note that $P^{(\theta)}_{Y^n|X^n}$ is a legitimate
distribution (like $P^{(\theta)}_{X|Y}$ defined in Theorem~\ref{improved-pv}). 
As a result, the new denominator of the fraction in \eqref{nti} 
(i.e., $\sum_{\hat x^n\in\xc^n}P_{X^n}(\hat x^n)P^{(\theta)}_{Y^n|X^n}(y^n|x^n)$)
is a true distribution on $\yc^n$; it is indeed the distribution of
the output due to an input with distribution $P_{X^n}$ sent over 
a channel with (legitimate) tilted statistics $P_{Y^n|X^n}^{(\theta)}$.
We thus conclude that for channels satisfying the invariance condition
of \eqref{invariance-cond},
the upper bounds for the channel reliability function 
in \eqref{rebound-theta} and \eqref{rebound} are actually based on 
the channel information density $i_{X^nW^n}^{(\theta)}(x^n;y^n)$ of an 
{\it auxiliary channel} whose transition probability $P_{Y^n|X^n}^{(\theta)}$
is the \emph{tilted} counterpart of the original channel transition probability $P_{Y^n|X^n}$.

When the output alphabet is finite, 
the channel $\Wb$ satisfies \eqref{invariance-cond} if it is {\it row-symmetric}, i.e.,
if the rows of its transition matrix $[p_{x^ny^n}]$
of size $|\xc^n| \times|\yc^n|$, where $p_{x^ny^n} \triangleq P_{Y^n|X^n}(y^n|x^n)$,
are permutations of each other for each $n$. 
Note that channels whose transition matrix $[p_{x^ny^n}]$ is symmetric in the Gallager sense 
\cite[p.~94]{gallager} for each $n$ are row-symmetric; such channels include
the memoryless BSC and BEC.

When the output alphabet is continuous (i.e., with $\yc=\mathbb{R}$) 
and the channel is described by a 
sequence of $n$-dimensional transition (conditional) probability 
density functions (pdfs) $f_{Y^n|X^n}$, the invariance
condition of \eqref{invariance-cond} translates into
\begin{equation}\label{invariance-cond-cont}
\int_{\hat y^n\in\mathbb{R}^n}f_{Y^n|X^n}^\theta(\hat y^n|x^n) 
d{\hat y}_1 \cdots d{\hat y}_n
= \int_{\hat y^n\in\mathbb{R}^n}f_{Y^n|X^n}^\theta(\hat y^n|\hat x^n) 
d{\hat y}_1 \cdots d{\hat y}_n
\end{equation}
$\forall \ x^n, {\hat x^n} \in \xc^n$, $n=1,2,\cdots$.
The memoryless finite-input AWGN channel 
and the memoryless binary-input (with $\xc=\{-1,+1\}$)
{\it output-symmetric} channel, i.e., whose transition pdf
satisfies $f_{Y|X}(y|-1)=f_{Y|X}(-y|+1)$ $\forall \ y \in \mathbb{R}$,
fulfill \eqref{invariance-cond-cont}.
}
\end{obs}

\bigskip

\begin{obs}{\rm \ It can be shown along similar lines as the proof
of \cite[Theorem~1]{pv02} that one can interchange 
the supremum and limit infimum (over $n$)
in $E_{\text{PV}}^{(\theta)}(R)$ and $\bar E_{\text{PV}}(R)$ and obtain
\begin{eqnarray}\label{underbarE}
\lim_{\gamma\downarrow 0}\underbar{$E$}_{\text{PV}}^{(\theta)}(R+\gamma)
\leq E_{\text{PV}}^{(\theta)}(R)
\leq \underbar{$E$}_{\text{PV}}^{(\theta)}(R)
\quad\text{and}\quad
\lim_{\gamma\downarrow 0}\underbar{$E$}_{\text{PV}}(R+\gamma)
\leq \bar{E}_{\text{PV}}(R)
\leq \underbar{$E$}_{\text{PV}}(R),
\end{eqnarray}
where
\begin{multline}
\underbar{$E$}_{\text{PV}}(R)\triangleq
\liminf_{n\rightarrow\infty}\sup_{X^n\in{\cal Q}_n(R)}\lim_{\theta\rightarrow\infty}-\frac 1n\log
P_{X^nW^n}\bigg\{(x^n,y^n)\in\xc^n\times\yc^n:\\
\frac{1}{n}j_{X^nW^n}^{(\theta)}(x^n;y^n)\leq R
\bigg\},\nonumber
\end{multline}
\begin{eqnarray*}
\underbar{$E$}_{\text{PV}}^{(\theta)}(R)\triangleq\liminf_{n\rightarrow\infty}
 \sup_{X^n\in{\cal Q}_n(R)}-\frac 1n\log
P_{X^nW^n}\left\{(x^n,y^n)\in\xc^n\times\yc^n: \frac{1}{n}j_{X^nW^n}^{(\theta)}(x^n;y^n)
\leq R\right\} \nonumber
\end{eqnarray*}
and
\begin{multline}
\qc_n(R)\triangleq\bigg\{X^n:
P_{X^n}(x^n)=\frac 1{|\scc(X^n)|}\text{ for }x^n\in\scc(X^n)
\ \text{and}\ R <\frac
1n\log|\scc(X^n)|\bigg\}.\nonumber
\end{multline}
The new expressions that take the supremum over $\qc_n(R)$ before letting $n$ approaching 
infinity provide an alternative possibility for the evaluation of the two bounds. 
In particular, $\qc_n(R)$ becomes a finite set as the input alphabet is finite;
hence, taking the supremum over $\qc_n(R)$ can be replaced with a maximization operation.
Inequality \eqref{underbarE} nevertheless implies that 
 $E_{\text{PV}}^{(\theta)}(R)
=\underbar{$E$}_{\text{PV}}^{(\theta)}(R)$ and
$\bar{E}_{\text{PV}}(R)=\underbar{$E$}_{\text{PV}}(R)$ almost everywhere in $R$
(since these functions are non-increasing in $R$).
}
\end{obs}

\subsection{Information-Spectral Characterization of the Reliability Function for a 
Class of Channels}\label{exact-exp}

We next employ Theorem~\ref{tightness-theo} to show that the
upper bound in \eqref{rebound} is tight for the memoryless finite-input AWGN channel as
well as a larger class of channels, hence providing an information-spectral
characterization
for the reliability function of these channels. This exact expression $E^\ast(R)=\bar{E}_{\text{PV}}(R)$
holds for all rates $R$ (below channel capacity), 
albeit its determination in single-letter form
(i.e., solving the optimization of a large-deviation rate function)
remains a challenging open problem.
  
We first focus on the Gaussian channel and then present the result for a wider class of channels.
Consider a finite-input AWGN channel described by $Y_i=X_i+Z_i$, $i=1,2,\cdots$, 
where $X_i$, $Y_i$ and $Z_i$ are the channel's input, output and noise at time $i$,
respectively. 
We assume that the noise process $\Zb$ is i.i.d.\ with each $Z_i$ being a zero-mean
Gaussian random variable with variance $\sigma^2>0$. We also assume that the noise
and input processes are independent from each other.

\bigskip

\begin{theorem}\label{awgn-exact-exp-theo}
{\rm
The channel reliability function $E^\ast(R)$ of the above finite-input
AWGN channel satisfies
\begin{eqnarray*}
E^\ast(R) &=& \bar{E}_{\text{PV}}(R) \\
&=& \sup_{\boldsymbol X\in{\cal Q}(R)}\liminf_{n\rightarrow\infty}
\lim_{\theta\rightarrow\infty}-\frac 1n\log
P_{X^nW^n}\left\{(x^n,y^n)\in\xc^n\times\yc^n: \frac{1}{n}j_{X^nW^n}^{(\theta)}
(x^n;y^n)\leq R\right\} \\
&=& \sup_{\boldsymbol X\in{\cal Q}(R)}\liminf_{n\rightarrow\infty}
\lim_{\theta\rightarrow\infty}-\frac 1n\log
P_{X^nW^n}\left\{(x^n,y^n)\in\xc^n\times\yc^n: \frac{1}{n}i_{X^nW^n}^{(\theta)}
(x^n;y^n)\leq R\right\}
\end{eqnarray*}
for any $0<R<C$, where $C$ denotes the channel's capacity,
and $j_{X^nW^n}^{(\theta)}(x^n,y^n)$ and $i_{X^nW^n}^{(\theta)}(x^n,y^n)$
are given in \eqref{ti} and \eqref{nti}, respectively.
}
\end{theorem}

\medskip
\begin{proof}
Fix $0<R<C$. Let its channel input
$X^n$ be uniformly distributed over a codebook $\code_n \subset \xc^n$
and let $Y^n$ be the corresponding channel output. 
Then, for $x^n \in \code_n$,
\begin{eqnarray*}
P_{X^n|Y^n}(x^n|y^n)&=&\frac{P_{X^n}(x^n)f_{Y^n|X^n}(y^n|x^n)}{f_{Y^n}(y^n)}\\
&=&\frac 1{|\code_n|\cdot f_{Y^n}(y^n)}
\frac{1}{(2\pi\sigma^2)^{n/2}}\exp\left\{-\frac{\|y^n-x^n\|^2}{2\sigma^2}\right\},
\end{eqnarray*}
where $\|\cdot\|$ denotes the Euclidean norm.
For a given $y^n$ received at the channel output,  
if $\ell(y^n)$ as defined in \eqref{ell-def} is greater than or equal to 2,
then there exist distinct codewords $x^n$ and $\tilde x^n$ in $\code_n$
such that
$$\|y^n-x^n\|^2=\|y^n-\tilde x^n\|^2, \text{ equivalently }
\sum_{i=1}^n(x_i-\tilde x_i)y_i=\frac 12\sum_{i=1}^n(x_i^2-\tilde x_i^2);$$
hence such $y^n$ belongs to an (affine) hyperplane in $\mathbb{R}^n$.
In other words, we have that 
$$\{y^n\in\mathbb{R}^n:\ell(y^n)\ge 2 \} \subseteq\yc(\code_n),$$
where
$$\yc(\code_n)\triangleq\left\{y^n\in\mathbb{R}^n:\|y^n-x^n\|^2=\|y^n-\tilde x^n\|^2\text{ for some }x^n,\tilde x^n\in\code_n\text{ and }x^n\neq \tilde x^n\right\}$$
consists of the union of $\binom{|\code_n|}{2}$ hyperplanes
in $\mathbb{R}^n$. But as the Lebesgue measure of every hyperplane in $\mathbb{R}^n$
is zero (since its volume is zero), we then obtain that the above finite union of hyperplanes
has Lebesgue measure zero. Thus, $P_{Y^n}\{\yc(\code_n)\}=0$ which directly yields
that $\Pr[\ell(Y^n)\ge 2]=0$, and hence $\Pr[\ell(Y^n)=1]=1$.
Theorem~\ref{tightness-theo} then implies that
\begin{equation*}
P_e(\code_n)=\lim_{\theta\rightarrow\infty}
P_{X^nW^n}\left\{(x^n,y^n)\in\xc^n\times\yc^n:j_{X^nW^n}^{(\theta)}(x^n;y^n)\leq \log M+\log \alpha\right\}
\end{equation*}
for $\alpha\in[0,1)$.
As a result, with $\alpha=e^{-n\gamma}$ for arbitrarily small $\gamma>0$,
\begin{eqnarray*}
&& \liminf_{n\rightarrow\infty}-\frac 1n\log P_e(\code_n) \\
&& \hspace{-0.15in}  =\liminf_{n\rightarrow\infty}\lim_{\theta\rightarrow\infty}-\frac 1n\log
P_{X^nW^n}\left\{(x^n,y^n)\in\xc^n\times\yc^n: 
\frac 1nj_{X^nW^n}^{(\theta)}(x^n;y^n)\leq \frac 1n\log |\code_n|-\gamma\right\},
\end{eqnarray*}
where $j_{X^nW^n}^{(\theta)}(x^n,y^n)$ is as defined in \eqref{ti}.
As stated in the proof of Theorem~\ref{exp-upper-bound}, 
the channel input that achieves the channel reliability
should has the chosen $\gamma$ and supports satisfying $\liminf_{n\rightarrow\infty}\frac 1n\log|\scc(X^n)|-\gamma$ strictly larger but arbitrarily close to $R$.
This concludes to
\begin{eqnarray*}
E^\ast(R)&=&\sup_{\{\codeb_n=\scc(X^n)\}_{n\geq 1}:\Xb\in\qc(R)}
\liminf_{n\rightarrow\infty}-\frac 1n\log P_e(\code_n)\\
&=&\sup_{\Xb\in\qc(R)}\liminf_{n\rightarrow\infty}\lim_{\theta\rightarrow\infty}
-\frac 1n\log P_{X^nW^n}\left\{(x^n,y^n)\in\xc^n\times\yc^n:
\frac 1nj_{X^nW^n}^{(\theta)}(x^n;y^n)\leq R\right\} \\
&\triangleq&\bar{E}_{\text{PV}}(R).
\end{eqnarray*}

Furthermore, since this channel satisfies \eqref{invariance-cond-cont},
we can replace $j_{X^nW^n}^{(\theta)}(x^n;y^n)$ with $i_{X^nW^n}^{(\theta)}(x^n;y^n)$
in the expression of $\bar{E}_{\text{PV}}(R)$ as shown
in Observation~\ref{auxiliary-channel} to obtain that
\begin{eqnarray*}
E^\ast(R)=
\sup_{\Xb\in\qc(R)}\liminf_{n\rightarrow\infty}\lim_{\theta\rightarrow\infty}
-\frac 1n\log P_{X^nW^n}\left\{(x^n,y^n)\in\xc^n\times\yc^n:
\frac 1ni_{X^nW^n}^{(\theta)}(x^n;y^n)\leq R\right\}.
\end{eqnarray*}
\end{proof}

An information-spectral representation of $E^\ast(R)$ for the memoryless finite-input AWGN channel
is thus established for all rates, although its solution in closed (single-letter)
form is still a daunting task.

We emphasize that the above finding also holds for any channel satisfying
$\ell(Y^n)=1$ almost surely in $P_{Y^n}$ as shown above; we hence have
the following result (which directly follows from Theorem~\ref{tightness-theo} 
along the same lines as the above proof).

\bigskip

\begin{corollary}\label{exact-exp-cor}{\rm
Given a channel $\Wb$, if for its input $\Xb$ uniform over any block codebook $\code_n$,
the following holds almost surely in $P_{Y^n}$
\begin{equation}\label{cor1-condition}
\max_{x^n\in\codeb_n}P_{Y^n|X^n}(y^n|x^n)>
\max_{x^n\in\codeb_n\setminus\{e(y^n)\}}P_{Y^n|X^n}(y^n|x^n)
\end{equation}
for each $n=1,2,\cdots$,
where $e_{ML}(y^n)=\arg\max_{x^n\in\codeb_n}P_{Y^n|X^n}(y^n|x^n)$ is the maximum likelihood
estimate of the transmitted codeword from the received channel output $y^n$, 
then the channel reliability function of $\Wb$ is given by
\begin{eqnarray*}
E^\ast(R) &=& \bar{E}_{\text{PV}}(R) \\
&=& \sup_{\boldsymbol X\in{\cal Q}(R)}\liminf_{n\rightarrow\infty}
\lim_{\theta\rightarrow\infty}-\frac 1n\log
P_{X^nW^n}\left\{(x^n,y^n)\in\xc^n\times\yc^n: \frac{1}{n}j_{X^nW^n}^{(\theta)}(x^n;y^n)\leq R
\right\} 
\end{eqnarray*}
for any $0<R<C$, where $C$ is the channel's capacity.

Furthermore, if the channel satisfies the invariance conditions \eqref{invariance-cond}
or \eqref{invariance-cond-cont}, then 
$j_{X^nW^n}^{(\theta)}(x^n;y^n) = i_{X^nW^n}^{(\theta)}(x^n;y^n)$, which is the information density
for the auxiliary channel with transition distribution $P_{Y^n|X^n}^{(\theta)}$ (i.e., the
tilted distribution of the original channel distribution $P_{Y^n|X^n}$). In this case
the channel reliability function becomes
\begin{eqnarray*}
E^\ast(R) &=& \bar{E}_{\text{PV}}(R) \\
&=& \sup_{\boldsymbol X\in{\cal Q}(R)}\liminf_{n\rightarrow\infty}
\lim_{\theta\rightarrow\infty}-\frac 1n\log
P_{X^nW^n}\left\{(x^n,y^n)\in\xc^n\times\yc^n: \frac{1}{n}i_{X^nW^n}^{(\theta)}(x^n;y^n)\leq R
\right\} 
\end{eqnarray*}
for any $0<R<C$.
}
\end{corollary}

\bigskip

\begin{obs}{\rm \ 
Corollary \ref{exact-exp-cor} requires condition \eqref{cor1-condition} to be valid
for any block codebook $\code_n$ and for each $n=1,2,\cdots$. One can immediately weaken the condition by considering only sufficiently large $n$; but without further knowledge on the optimal codebook (equivalently, the optimal channel input $\Xb$ that 
achieves $\bar{E}_{\text{PV}}(R)$), it may be hard to derive an alternative condition
for \eqref{cor1-condition} that holds unanimously for \emph{any} codebook. 
In particular, for discrete memoryless channels (DMC) with finite or countably infinite 
output alphabets, a codebook that fails condition~\eqref{cor1-condition} can always 
be constructed except if the channels are not noiseless (i.e., perfect).\footnote{
As a simple proof, note that for a noisy DMC
there exist two inputs $a,a'\in\xc$ and an output $b\in\yc$ satisfying
$\min\{P_{Y|X}(b|a),P_{Y|X}(b|a')\}>0$. Then for a codebook $\code_n$ consisting of two distinct codewords $x^n$ and $\tilde x^n$, where one of them is the permutation of the other, and their components are either $a$ or $a'$, we obtain 
\begin{eqnarray*}
P_{X^n|Y^n}(x^n|y^n)=\frac{P_{X^n}(x^n)P_{Y^n|X^n}(y^n|x^n)}{P_{Y^n}(y^n)}
=\frac{P_{Y^n|X^n}(y^n|x^n)}{|\code_n|\cdot P_{Y^n}(y^n)}
=\frac{P_{X^n}(\tilde x^n)P_{Y^n|X^n}(y^n|\tilde x^n)}{P_{Y^n}(y^n)}
=P_{X^n|Y^n}(\tilde x^n|y^n)
\end{eqnarray*}
for the channel output $y^n$ satisfying $y_i=b$ for every $1\leq i\leq n$;
hence, $\ell(y^n)\geq 2$ with  $P_{Y^n}(y^n)
=\frac 12P_{Y^n|X^n}(y^n|x^n)+\frac 12P_{Y^n|X^n}(y^n|\tilde x^n)>0$. This codebook 
therefore violates condition \eqref{cor1-condition}.

Notably, for a channel satisfying $\min\{P_{Y|X}(b|a),P_{Y|X}(b|a')\}=0$
for every unequal $a,a'\in\xc$ and $b\in\yc$, the error rate is zero for any codebook $\code_n$. So, only under such a noiseless situation can the finite- or countable-output DMC meet the strict requirement that $\ell(Y^n)=1$ with probability one for any codebook $\code_n$.
}
Hence, in its current form, Corollary~\ref{exact-exp-cor} is not useful for
discrete-output channels; instead, it is of interest for continuous-output channels. 
}
\end{obs}

\bigskip

\begin{obs}{\rm \ In light of the above observation, we further
consider channels with continuous-output alphabets.
For a channel that admits a channel transition pdf, the proof of 
Theorem~\ref{awgn-exact-exp-theo} actually indicates that as long as 
$P_{Y^n}\{\yc(\code_n)\}=0$ for any block codebook $\code_n$, where
$$\yc(\code_n)\triangleq\left\{y^n\in\mathbb{R}^n:
f_{Y^n|X^n}(y^n|x^n)=f_{Y^n|X^n}(y^n|\tilde x^n)\text{ for some }x^n,\tilde x^n\in\code_n\text{ and }x^n\neq \tilde x^n\right\},$$
we have $\Pr[\ell(Y^n)=1]=1$ and \eqref{cor1-condition} holds.
We note that this is indeed valid for any sequence of transition pdf's for which
the number of solutions in $y_n$ satisfying
$$f_{Y^n|X^n}(y^n|x^n)=f_{Y^n|X^n}(y^n|\tilde x^n)$$ for given codewords $x^n$, $\tilde x^n$
in $\code_n$ and given $y^{n-1}$ is either finite or countable (as this condition
immediately implies that $\yc(\code_n)$ has Lebesgue measure zero).
A large class of channels satisfy this condition. 
For example, channels with memoryless additive noise, where the noise
pdf is not uniform or piecewise-uniform, satisfy this condition and hence
\eqref{cor1-condition} and Corollary~\ref{exact-exp-cor}.
This allows for most standard continuous distributions for the noise,
such as the generalized-Gaussian distribution 
with shape parameter $c>0$ (e.g., cf. \cite{miller-thomas});
this distribution includes the Gaussian and Laplacian distributions as special
cases, realized for $c=2$ and $c=1$, respectively.

}
\label{cont-channel-obs}
\end{obs}

\subsection{Examples of Channels for which the $E_{\text{PV}}^{(\theta)}(R)$ Bound Is Not Tight} 
As already mentioned, the (analytical or numerical) computation of both upper bounds, 
$E_{\text{PV}}^{(\theta)}(R)$ and $\bar{E}_{\text{PV}}(R)$, 
to the channel reliability function, given in \eqref{rebound-theta} and \eqref{rebound},
respectively, is formidable since they involve a difficult supremum operation
of input processes in $\qc(R)$ in addition to the limit operations.

We can however lower-bound $E_{\text{PV}}^{(\theta)}(R)$, for a given (fixed) $\theta$,
using an auxiliary class of i.i.d.\ inputs and compare this lower bound
to $E_{\text{PV}}^{(\theta)}(R)$ with familiar channel reliability function upper bounds
(such as the sphere-packing upper bound). If the former is shown to be strictly
larger than the latter for a range of rates, then this indicates
that for that particular $\theta$, $E_{\text{PV}}^{(\theta)}(R)$ is not tight.
The lower bound to $E_{\text{PV}}^{(\theta)}(R)$, which we denote by $F(R,\theta)$,
is derived in Appendix~B and given in \eqref{F-R-theta} for
the case of memoryless channels.
We herein calculate $F(R,\theta)$ numerically to demonstrate
that $E_{\text{PV}}^{(\theta)}(R)$ is not tight within a rate range and for certain
choices of $\theta$ (including $\theta=1$ which gives the Poor-Verd\'{u} bound
of \eqref{tightened-pv-exp}); this is shown for two standard binary-input memoryless channels:
the BSC and the Z-channel.

\subsubsection{Memoryless BSC}
For the BSC with crossover probability $\ep$, setting $p \triangleq P_{\bar X}(1)$
and $s=\frac{1}{1-\rho}$ in \eqref{F-R-theta} yields
\begin{eqnarray*}
E_{\text{PV}}^{(\theta)}(R) &\geq & F(R,\theta) \\
&=& \sup_{0<s<1}\bigg\{\left(1-\frac 1s\right)R-\inf_{p: h_{\rm b}(p)>R}\log
\bigg[\frac{(1-p)(1-\ep)^{1+\theta-\theta/s}+p\ep^{1+\theta-\theta/s}}{\left[(1-p)(1-\ep)
^\theta+p\ep^\theta\right]^{(1-1/s)}} \\
&& \hspace{2.5in}  +\frac{(1-p)\ep^{1+\theta-\theta/s}+p(1-\ep)^{1+\theta-\theta/s}}
{\left[(1-p)\ep^\theta+p(1-\ep)^\theta\right]^{(1-1/s)}}\bigg]\bigg\}
\end{eqnarray*}
for reals $\theta\ge 1$ and $0<R<C=\log(2) - h_{\rm b}(\ep)$, where $C$ is the channel capacity
and $h_{\rm b}(\ep)=-\ep\log \ep - (1-\ep) \log (1-\ep)$ is the binary entropy function.

We compare $F(R,\theta)$ with the sphere packing upper bound 
to the BSC's reliability function (e.g., \cite{gallager,blahut}), which is 
denoted by $E_{\text{sp}}(R)$ and given by
\begin{equation*}
E_{\text{sp}}(R)=\sup_{0<s\leq 1}\left\{
    \left(1-\frac 1{s}\right)\left(R-\log 2\right)-\frac 1{s}\log\left[(1-\ep)^{s}+\ep^{s}
\right]\right\}
\end{equation*}
for $0<R<C$.
In Fig.~\ref{fig-bsc0.01}, we plot $E_{\text{sp}}(R)$ and $F(R,\theta)$ for $\theta=1$ and 2
and $\ep=0.01$.
The figure indicates that for $\theta=1$, $F(R,\theta) > E_{\text{sp}}(R)$ for all rates $R$.
This directly implies that 
$$E_{\text{PV}}(R)=E_{\text{PV}}^{(\theta=1)}(R) \ge F(R,\theta) > E_{\text{sp}}(R)$$
for all $0<R<C$. 
Now recall that the sphere-packing upper bound $E_{\text{sp}}(R)$ is loose at low rates (for rates $R$ less
than the critical rate \cite{gallager}) and tight (i.e., exactly equal to the channel reliability
function $E^\ast(R)$) at high rates (rates between the critical rate and capacity).
Thus for the BSC, the Poor-Verd\'{u} bound
of \eqref{tightened-pv-exp} is not tight for all rates.
Furthermore, note from the figure that since  
$F(R,\theta) < E_{\text{sp}}(R)$ for $\theta=2$, we cannot make a conclusion regarding
the tightness of $E_{\text{PV}}^{(\theta)}(R)$ in this case (this is also observed for $\theta > 2$).

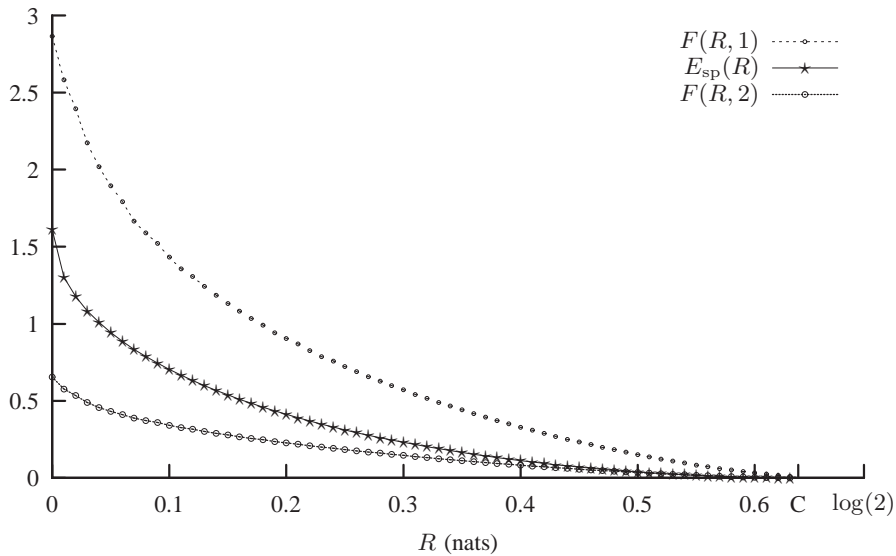
\begin{figure}[h!]
\begin{center}
\setlength{\unitlength}{0.120450pt}
\begin{picture}(3000,1800)(0,0)
\footnotesize
\thicklines \path(348,265)(389,265)
\put(307,265){\makebox(0,0)[r]{ 0}}
\thicklines \path(348,507)(389,507)
\put(307,507){\makebox(0,0)[r]{ 0.5}}
\thicklines \path(348,749)(389,749)
\put(307,749){\makebox(0,0)[r]{ 1}}
\thicklines \path(348,992)(389,992)
\put(307,992){\makebox(0,0)[r]{ 1.5}}
\thicklines \path(348,1234)(389,1234)
\put(307,1234){\makebox(0,0)[r]{ 2}}
\thicklines \path(348,1476)(389,1476)
\put(307,1476){\makebox(0,0)[r]{ 2.5}}
\thicklines \path(348,1718)(389,1718)
\put(307,1718){\makebox(0,0)[r]{ 3}}
\thicklines \path(348,265)(348,306)
\put(348,182){\makebox(0,0){0}}
\thicklines \path(716,265)(716,306)
\put(716,182){\makebox(0,0){0.1}}
\thicklines \path(1083,265)(1083,306)
\put(1083,182){\makebox(0,0){0.2}}
\thicklines \path(1451,265)(1451,306)
\put(1451,182){\makebox(0,0){0.3}}
\thicklines \path(1819,265)(1819,306)
\put(1819,182){\makebox(0,0){0.4}}
\thicklines \path(2187,265)(2187,306)
\put(2187,182){\makebox(0,0){0.5}}
\thicklines \path(2554,265)(2554,306)
\put(2554,182){\makebox(0,0){0.6}}
\thicklines \path(2691,265)(2691,306)
\put(2691,182){\makebox(0,0){C}}
\thicklines \path(2897,265)(2897,306)
\put(2897,182){\makebox(0,0){$\log(2)$}}
\thicklines \path(348,1718)(348,265)(2897,265)
\put(1622,58){\makebox(0,0){$R$ (nats)}}
\put(2569,1636){\makebox(0,0)[r]{$F(R,1)$}}
\thinlines \dashline[90]{10}(2610,1636)(2815,1636)
\thinlines \dashline[90]{10}(348,1653)(348,1653)(385,1516)(422,1425)(458,1318)(495,1243)(532,1183)(569,1133)(605,1072)(642,1035)(679,1002)(716,959)(753,922)(789,898)(826,867)(863,839)(900,813)(936,789)(973,766)(1010,745)(1047,721)(1083,703)(1120,686)(1157,665)(1194,646)(1231,632)(1267,615)(1304,599)(1341,583)(1378,569)(1414,555)(1451,542)(1488,527)(1525,515)(1562,502)(1598,491)(1635,479)(1672,467)(1709,456)(1745,445)(1782,435)(1819,424)(1856,415)(1893,405)(1929,395)(1966,386)(2003,378)(2040,370)(2076,361)(2113,354)(2150,346)
\thinlines \dashline[90]{10}(2150,346)(2187,338)(2223,331)(2260,324)(2297,318)(2334,312)(2371,305)(2407,300)(2444,294)(2481,290)(2518,285)(2554,281)(2591,277)(2628,273)(2665,270)
\put(348,1653){\circle{12}}
\put(385,1516){\circle{12}}
\put(422,1425){\circle{12}}
\put(458,1318){\circle{12}}
\put(495,1243){\circle{12}}
\put(532,1183){\circle{12}}
\put(569,1133){\circle{12}}
\put(605,1072){\circle{12}}
\put(642,1035){\circle{12}}
\put(679,1002){\circle{12}}
\put(716,959){\circle{12}}
\put(753,922){\circle{12}}
\put(789,898){\circle{12}}
\put(826,867){\circle{12}}
\put(863,839){\circle{12}}
\put(900,813){\circle{12}}
\put(936,789){\circle{12}}
\put(973,766){\circle{12}}
\put(1010,745){\circle{12}}
\put(1047,721){\circle{12}}
\put(1083,703){\circle{12}}
\put(1120,686){\circle{12}}
\put(1157,665){\circle{12}}
\put(1194,646){\circle{12}}
\put(1231,632){\circle{12}}
\put(1267,615){\circle{12}}
\put(1304,599){\circle{12}}
\put(1341,583){\circle{12}}
\put(1378,569){\circle{12}}
\put(1414,555){\circle{12}}
\put(1451,542){\circle{12}}
\put(1488,527){\circle{12}}
\put(1525,515){\circle{12}}
\put(1562,502){\circle{12}}
\put(1598,491){\circle{12}}
\put(1635,479){\circle{12}}
\put(1672,467){\circle{12}}
\put(1709,456){\circle{12}}
\put(1745,445){\circle{12}}
\put(1782,435){\circle{12}}
\put(1819,424){\circle{12}}
\put(1856,415){\circle{12}}
\put(1893,405){\circle{12}}
\put(1929,395){\circle{12}}
\put(1966,386){\circle{12}}
\put(2003,378){\circle{12}}
\put(2040,370){\circle{12}}
\put(2076,361){\circle{12}}
\put(2113,354){\circle{12}}
\put(2150,346){\circle{12}}
\put(2187,338){\circle{12}}
\put(2223,331){\circle{12}}
\put(2260,324){\circle{12}}
\put(2297,318){\circle{12}}
\put(2334,312){\circle{12}}
\put(2371,305){\circle{12}}
\put(2407,300){\circle{12}}
\put(2444,294){\circle{12}}
\put(2481,290){\circle{12}}
\put(2518,285){\circle{12}}
\put(2554,281){\circle{12}}
\put(2591,277){\circle{12}}
\put(2628,273){\circle{12}}
\put(2665,270){\circle{12}}
\put(2712,1636){\circle{12}}
\put(2569,1553){\makebox(0,0)[r]{$E_{\rm sp}(R)$}}
\thinlines \path(2610,1553)(2815,1553)
\thinlines \path(348,1046)(348,1046)(385,895)(422,835)(458,790)(495,754)(532,722)(569,694)(605,669)(642,647)(679,626)(716,606)(753,588)(789,571)(826,555)(863,540)(900,526)(936,512)(973,500)(1010,487)(1047,476)(1083,465)(1120,454)(1157,444)(1194,434)(1231,425)(1267,416)(1304,407)(1341,399)(1378,391)(1414,384)(1451,377)(1488,370)(1525,363)(1562,357)(1598,351)(1635,345)(1672,339)(1709,334)(1745,329)(1782,324)(1819,319)(1856,315)(1893,311)(1929,306)(1966,303)(2003,299)(2040,296)(2076,292)(2113,289)(2150,286)
\thinlines \path(2150,286)(2187,284)(2223,281)(2260,279)(2297,277)(2334,275)(2371,273)(2407,271)(2444,270)(2481,268)(2518,267)(2554,266)(2591,266)(2628,265)(2665,265)
\put(348,1046){\makebox(0,0){$\star$}}
\put(385,895){\makebox(0,0){$\star$}}
\put(422,835){\makebox(0,0){$\star$}}
\put(458,790){\makebox(0,0){$\star$}}
\put(495,754){\makebox(0,0){$\star$}}
\put(532,722){\makebox(0,0){$\star$}}
\put(569,694){\makebox(0,0){$\star$}}
\put(605,669){\makebox(0,0){$\star$}}
\put(642,647){\makebox(0,0){$\star$}}
\put(679,626){\makebox(0,0){$\star$}}
\put(716,606){\makebox(0,0){$\star$}}
\put(753,588){\makebox(0,0){$\star$}}
\put(789,571){\makebox(0,0){$\star$}}
\put(826,555){\makebox(0,0){$\star$}}
\put(863,540){\makebox(0,0){$\star$}}
\put(900,526){\makebox(0,0){$\star$}}
\put(936,512){\makebox(0,0){$\star$}}
\put(973,500){\makebox(0,0){$\star$}}
\put(1010,487){\makebox(0,0){$\star$}}
\put(1047,476){\makebox(0,0){$\star$}}
\put(1083,465){\makebox(0,0){$\star$}}
\put(1120,454){\makebox(0,0){$\star$}}
\put(1157,444){\makebox(0,0){$\star$}}
\put(1194,434){\makebox(0,0){$\star$}}
\put(1231,425){\makebox(0,0){$\star$}}
\put(1267,416){\makebox(0,0){$\star$}}
\put(1304,407){\makebox(0,0){$\star$}}
\put(1341,399){\makebox(0,0){$\star$}}
\put(1378,391){\makebox(0,0){$\star$}}
\put(1414,384){\makebox(0,0){$\star$}}
\put(1451,377){\makebox(0,0){$\star$}}
\put(1488,370){\makebox(0,0){$\star$}}
\put(1525,363){\makebox(0,0){$\star$}}
\put(1562,357){\makebox(0,0){$\star$}}
\put(1598,351){\makebox(0,0){$\star$}}
\put(1635,345){\makebox(0,0){$\star$}}
\put(1672,339){\makebox(0,0){$\star$}}
\put(1709,334){\makebox(0,0){$\star$}}
\put(1745,329){\makebox(0,0){$\star$}}
\put(1782,324){\makebox(0,0){$\star$}}
\put(1819,319){\makebox(0,0){$\star$}}
\put(1856,315){\makebox(0,0){$\star$}}
\put(1893,311){\makebox(0,0){$\star$}}
\put(1929,306){\makebox(0,0){$\star$}}
\put(1966,303){\makebox(0,0){$\star$}}
\put(2003,299){\makebox(0,0){$\star$}}
\put(2040,296){\makebox(0,0){$\star$}}
\put(2076,292){\makebox(0,0){$\star$}}
\put(2113,289){\makebox(0,0){$\star$}}
\put(2150,286){\makebox(0,0){$\star$}}
\put(2187,284){\makebox(0,0){$\star$}}
\put(2223,281){\makebox(0,0){$\star$}}
\put(2260,279){\makebox(0,0){$\star$}}
\put(2297,277){\makebox(0,0){$\star$}}
\put(2334,275){\makebox(0,0){$\star$}}
\put(2371,273){\makebox(0,0){$\star$}}
\put(2407,271){\makebox(0,0){$\star$}}
\put(2444,270){\makebox(0,0){$\star$}}
\put(2481,268){\makebox(0,0){$\star$}}
\put(2518,267){\makebox(0,0){$\star$}}
\put(2554,266){\makebox(0,0){$\star$}}
\put(2591,266){\makebox(0,0){$\star$}}
\put(2628,265){\makebox(0,0){$\star$}}
\put(2665,265){\makebox(0,0){$\star$}}
\put(2712,1553){\makebox(0,0){$\star$}}
\put(2569,1470){\makebox(0,0)[r]{$F(R,2)$}}
\thinlines \dottedline{10}(2610,1470)(2815,1470)
\thinlines \dottedline{10}(348,582)(348,582)(385,544)(422,524)(458,502)(495,486)(532,474)(569,464)(605,453)(642,445)(679,439)(716,430)(753,423)(789,418)(826,411)(863,405)(900,400)(936,394)(973,389)(1010,385)(1047,379)(1083,375)(1120,371)(1157,366)(1194,362)(1231,358)(1267,354)(1304,350)(1341,346)(1378,343)(1414,339)(1451,336)(1488,332)(1525,329)(1562,325)(1598,322)(1635,319)(1672,316)(1709,313)(1745,310)(1782,308)(1819,305)(1856,302)(1893,300)(1929,297)(1966,295)(2003,292)(2040,290)(2076,288)(2113,286)(2150,284)
\thinlines \dottedline{10}(2150,284)(2187,282)(2223,280)(2260,278)(2297,276)(2334,275)(2371,273)(2407,272)(2444,270)(2481,269)(2518,268)(2554,267)(2591,266)(2628,266)(2665,265)
\put(348,582){\circle{18}}
\put(385,544){\circle{18}}
\put(422,524){\circle{18}}
\put(458,502){\circle{18}}
\put(495,486){\circle{18}}
\put(532,474){\circle{18}}
\put(569,464){\circle{18}}
\put(605,453){\circle{18}}
\put(642,445){\circle{18}}
\put(679,439){\circle{18}}
\put(716,430){\circle{18}}
\put(753,423){\circle{18}}
\put(789,418){\circle{18}}
\put(826,411){\circle{18}}
\put(863,405){\circle{18}}
\put(900,400){\circle{18}}
\put(936,394){\circle{18}}
\put(973,389){\circle{18}}
\put(1010,385){\circle{18}}
\put(1047,379){\circle{18}}
\put(1083,375){\circle{18}}
\put(1120,371){\circle{18}}
\put(1157,366){\circle{18}}
\put(1194,362){\circle{18}}
\put(1231,358){\circle{18}}
\put(1267,354){\circle{18}}
\put(1304,350){\circle{18}}
\put(1341,346){\circle{18}}
\put(1378,343){\circle{18}}
\put(1414,339){\circle{18}}
\put(1451,336){\circle{18}}
\put(1488,332){\circle{18}}
\put(1525,329){\circle{18}}
\put(1562,325){\circle{18}}
\put(1598,322){\circle{18}}
\put(1635,319){\circle{18}}
\put(1672,316){\circle{18}}
\put(1709,313){\circle{18}}
\put(1745,310){\circle{18}}
\put(1782,308){\circle{18}}
\put(1819,305){\circle{18}}
\put(1856,302){\circle{18}}
\put(1893,300){\circle{18}}
\put(1929,297){\circle{18}}
\put(1966,295){\circle{18}}
\put(2003,292){\circle{18}}
\put(2040,290){\circle{18}}
\put(2076,288){\circle{18}}
\put(2113,286){\circle{18}}
\put(2150,284){\circle{18}}
\put(2187,282){\circle{18}}
\put(2223,280){\circle{18}}
\put(2260,278){\circle{18}}
\put(2297,276){\circle{18}}
\put(2334,275){\circle{18}}
\put(2371,273){\circle{18}}
\put(2407,272){\circle{18}}
\put(2444,270){\circle{18}}
\put(2481,269){\circle{18}}
\put(2518,268){\circle{18}}
\put(2554,267){\circle{18}}
\put(2591,266){\circle{18}}
\put(2628,266){\circle{18}}
\put(2665,265){\circle{18}}
\put(2712,1470){\circle{18}}
\thicklines \path(348,1718)(348,265)(2897,265)
\end{picture}
\end{center}
\vspace{-0.2in}
\caption{BSC with crossover probability $\ep=0.01$: 
lower bound $F(R,\theta)$ to $E_{\text{PV}}^{(\theta)}(R)$ for $\theta=1,2$
and the sphere packing bound $E_{\rm sp}(R)$.}
\label{fig-bsc0.01}
\end{figure}

\subsubsection{Memoryless Z-Channel}
We next consider the memoryless binary Z-channel 
described by $P_{Y|X}(0|1)=\ep$ and $P_{Y|X}(0|0)=1$. 
Again, setting $p \triangleq P_{\bar X}(1)$
and $s=\frac{1}{1-\rho}$ in \eqref{F-R-theta} yields
\begin{eqnarray*}
E_{\text{PV}}^{(\theta)}(R) &\geq & F(R,\theta) \\
&=& \sup_{0<s<1}\left\{\left(1-\frac 1s\right) R-\inf_{p: h_{\rm b}(p)>R}\log
\left[
\frac{1-p+p\ep^{1+\theta-\theta/s}}{\left[1-p+p\ep^{\theta}\right]^{1-1/s}}+
p^{1/s}(1-\ep)\right]\right\}
\end{eqnarray*}
for $\theta\ge 1$ and $0<R<C=\log\left(1+(1-\ep)\ep^{\frac{\ep}{1-\ep}}\right)$.
Furthermore, the channel's sphere packing upper bound is given by
\begin{equation*}
E_{\text{sp}}(R)=\sup_{0<s\leq 1}\left\{
    \left(1-\frac 1s\right) R-\inf_{0\leq p\leq 1}\log\left[\left(1-p+p\ep^{s}
\right)^{1/s}+p^{1/s}(1-\ep)\right]\right\}\end{equation*}
for $0<R<C$.
In Fig.~\ref{fig-z-channel0.01}, we plot $E_{\text{sp}}(R)$ and $F(R,\theta)$ for $\theta=1,3,10,100$
and $\ep=0.01$. We remark from the figure that for all considered values of $\theta$
(including $\theta$ very large not shown herein), $F(R,\theta) > E_{\text{sp}}(R)$ for high rates.
This leads us to conclude that for the Z-channel, bound 
$E_{\text{PV}}^{(\theta)}(R)$ of \eqref{rebound-theta}  is not tight
at high rates even when $\theta$ approaches infinity. 

\begin{figure}[h!]
\begin{center}
\input{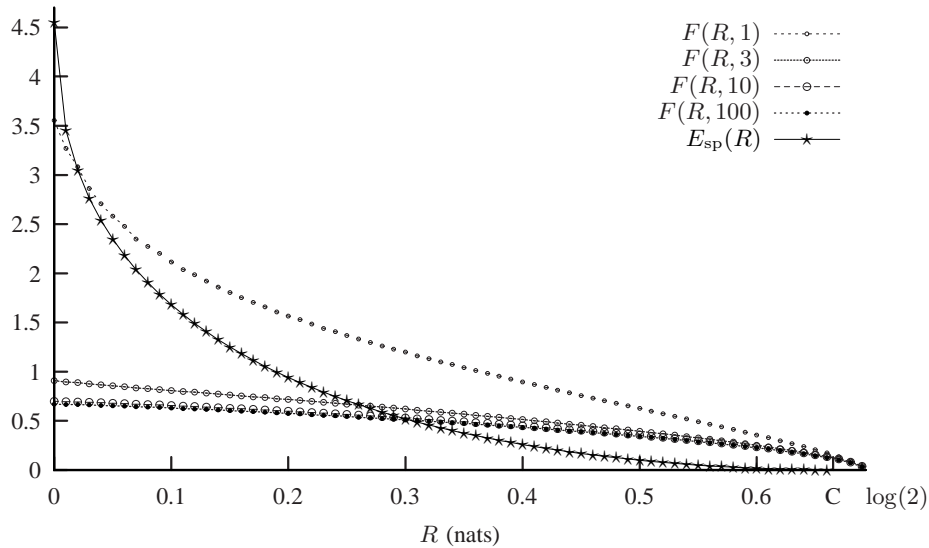}
\end{center}
\vspace{-0.2in}
\caption{Z-channel with crossover probability $\ep=0.01$: 
lower bound $F(R,\theta)$ to $E_{\text{PV}}^{(\theta)}(R)$ for $\theta=1,3,10,100$
and the sphere packing bound $E_{\rm sp}(R)$.}
\label{fig-z-channel0.01}
\end{figure}

\bigskip

\begin{obs}{\rm 
It should be emphasized that
the above numerical examples regarding the looseness of $E_{\text{PV}}^{(\theta)}(R)$ within
a rate region and for given values
of $\theta$ do not shed any light on the tightness of $\bar{E}_{\text{PV}}(R)$
given in \eqref{rebound}, since the expression of $\bar{E}_{\text{PV}}(R)$ requires
taking the limit with respect to $\theta$ {\it before} taking the limit with respect
to the blocklength $n$.
}
\end{obs}

\section{Conclusion}\label{conclusion}
In this work, we generalized the Poor-Verd\'{u} lower bound
for the multihypothesis testing error probability. The new bound, which involves
the tilted posterior distribution of the hypothesis given 
the observation with tilting parameter 
$\theta$, reduces to the original Poor-Verd\'{u} bound when $\theta=1$.
We established a sufficient condition 
under which the bound  (without its multiplicative factor)
provides the exact error probability when $\theta \rightarrow \infty$.
We also provided some examples to illustrate the tightness of the bound in terms
of $\theta$.

We next applied the new bound to obtain two new upper information-spectrum 
based bounds to the reliability function of general channels with memory,
$E_{\text{PV}}^{(\theta)}(R)$ and $\bar{E}_{\text{PV}}(R)$, given in \eqref{rebound-theta}
and \eqref{rebound}, respectively.
It was shown that $\bar{E}_{\text{PV}}(R)$ is tight at all rates (below channel capacity)
for a class of channels
that include the finite-input memoryless Gaussian channel, hence providing
an information-spectral characterization for these channels' reliability function.
The determination of $\bar{E}_{\text{PV}}(R)$ in closed form and its calculation
remains a challenging problem (specially at low rates) as it involves
taking the limit with respect to $\theta$ followed by optimizing
the resulting large-deviation rate function over a constrained set of input
processes (see \eqref{rebound}). It is anticipated that i.i.d.~channel inputs
are unlikely to be a valid optimizer for $\bar{E}_{\text{PV}}(R)$. Although
the evaluation of $\bar{E}_{\text{PV}}(R)$ for non--i.i.d.~channel inputs 
appears difficult, the judicious use of Markovian inputs might be worthwhile
investigating in the future. 


\section*{Appendix~A}
\label{limit}
\begin{lemma}
If the limit (in $n$) of $a_{n,j}$ exists for every $j=1,2,3,\ldots$,
then
$$\lim_{n\rightarrow\infty}\sum_{j=1}^\infty a_{n,j}
=\sum_{j=1}^\infty \lim_{n\rightarrow\infty} a_{n,j}.$$
\end{lemma}
\begin{proof}
Since for any sequences $\{b_n\}$ and $\{c_n\}$, 
$$\liminf_{n\rightarrow\infty}\left(b_n+c_n\right)
\geq \liminf_{n\rightarrow\infty}b_n+\liminf_{n\rightarrow\infty}c_n,$$
we recursively have that
\begin{eqnarray*}
\liminf_{n\rightarrow\infty}\sum_{j=1}^\infty a_{n,j}
&\geq& \liminf_{n\rightarrow\infty}a_{n,1}+\liminf_{n\rightarrow\infty}\sum_{j=2
}^\infty a_{n,j}\\
&\geq& \liminf_{n\rightarrow\infty}a_{n,1}+\liminf_{n\rightarrow\infty}a_{n,2}+
\liminf_{n\rightarrow\infty}\sum_{j=3}^\infty a_{n,j} \\ &\geq& \cdots\\
&\geq&
\sum_{j=1}^\infty \liminf_{n\rightarrow\infty}a_{n,j}.
\end{eqnarray*}
Similarly, since 
$$\limsup_{n\rightarrow\infty}\left(b_n+c_n,\right)
\leq \limsup_{n\rightarrow\infty}b_n+\limsup_{n\rightarrow\infty}c_n,$$
we obtain that
$$\limsup_{n\rightarrow\infty}\sum_{j=1}^\infty a_{n,j}
\leq \sum_{j=1}^\infty \limsup_{n\rightarrow\infty}a_{n,j}.$$
Since 
$$\limsup_{n\rightarrow\infty}a_{n,j}=\liminf_{n\rightarrow\infty}a_{n,j}=\lim_{
n\rightarrow\infty}a_{n,j}\mbox{ for every }j,$$
we have
$$\sum_{j=1}^\infty \lim_{n\rightarrow\infty}a_{n,j}\geq\limsup_{n\rightarrow\infty}
\sum_{j=1}^\infty a_{n,j}\geq
\liminf_{n\rightarrow\infty}\sum_{j=1}^\infty a_{n,j}\geq\sum_{j=1}^\infty \lim_
{n\rightarrow\infty}a_{n,j},$$
which immediately yields the desired result.
\end{proof}

\section*{Appendix~B}\label{exp-appendix}

We derive a lower bound to $E_{\text{PV}}^{(\theta)}(R)$ given in \eqref{rebound-theta},
which can be numerically evaluated for different values of $\theta$
when the channel is memoryless.


Consider a general channel $\Wb=\{W^n\}_{n=1}^\infty$ with finite input alphabet $\xc$ and
arbitrary output alphabet $\yc$.
Fix $R>0$. Given an i.i.d.~process ${\bar \Xb}= \{\bar X^n\}_{n=1}^\infty$ with 
alphabet $\xc$ and entropy $H(\bar X)>R$ 
and a constant $0<\delta<H(\bar X)-R$ arbitrarily small, define the (weakly) 
$\delta$-typical set as:
\begin{eqnarray*}
\fc_n(\de|\bar X)&\triangleq&
\left\{x^n\in\xc^n:\left|-\frac 1n
\log P_{\bar X^n}(x^n)-
H(\bar X)\right| \le \de\right\} \\
&=&
\left\{x^n\in\xc^n:\left|-\frac 1n\sum_{i=1}^n
\log P_{\bar X}(x_i)-
H(\bar X)\right| \le \de\right\}.
\end{eqnarray*}
We now recall the consequence of the Asymptotic Equipartition Property for
i.i.d.\ (memoryless) sources (e.g., see \cite{blahut,cover-thomas}).

\bigskip

\begin{prop}\label{aep}
{\rm 
Given an i.i.d.~source $\{\bar X_n\}_{n=1}^{\infty}$
with entropy $H(\bar X)$ and any $\de$ greater than zero,
then its $\de$-typical set $\fc_n(\de|\bar X)$ satisfies the following.
\begin{enumerate}
\item\label{sm1} If $x^n\in\fc_n(\de|\bar X)$, then
$e^{-n(H(\bar X)+\de)} \le P_{\bar X^n}(x^n) \le e^{-n(H(\bar X)-\de)}$.

\item\label{sm2} $P_{\bar X^n}\left(\fc_n^c(\de|\bar X)\right)<\de$ for sufficiently large $n$, where
the superscript ``$c$'' denotes the complement set operation.

\item\label{sm3} $|\fc_n(\de|\bar X)|>(1-\de)e^{n(H(\bar X)-\de)}$ for sufficiently large $n$,
and $|\fc_n(\de|\bar X)|\leq e^{n(H(\bar X)+\de)}$ for every $n$, where
$|\fc_n(\de|\bar X)|$ denotes the number of elements in $\fc_n(\de|\bar X)$.
\end{enumerate}
}\end{prop}

\bigskip

Let ${\hat \Xb}=\{\hat X^n\}_{n=1}^\infty$ be a process that is 
uniformly distributed over $\fc_n(\de|\bar X)$ for each $n$;
i.e., $P_{\hat X^n}(x^n)=\frac{1}{|\fc(\delta|\bar X)|}$ for $x^n\in\fc_n(\de|\bar X)$
and $n=1,2,\cdots$.
From Proposition~\ref{aep}, we also obtain that for $n$ sufficiently large and 
$x^n\in\fc_n(\de|\bar X)$,
\begin{eqnarray}\label{aep-ineq}
(1-\delta)e^{-2n\delta}\leq P_{\bar X^n}(x^n)|\fc_n(\delta|\bar X)|=\frac{P_{\bar X^n}(x^n)}{P_{\hat X^n}(x^n)}\leq e^{2n\delta}.
\end{eqnarray}
For ${\hat \Xb}$ to belong to the set $\qc(R)$ as defined in \eqref{input-set},
it is required
that
\begin{equation}\label{typical-condition}
\liminf_{n\rightarrow\infty}\frac 1n\log|\scc(\hat X^n)|
=\liminf_{n\rightarrow\infty}\frac 1n\log|\fc(\de|\bar X)|>R.
\end{equation}
But condition \eqref{typical-condition} can be guaranteed by setting $H(\bar X)>R$
and taking $\delta<H(\bar X)-R$ (as already assumed) since 
$$\liminf_{n\rightarrow\infty}\frac 1n\log|\fc(\de|\bar X)|
\geq \liminf_{n\rightarrow\infty}\frac 1n\log (1-\de)e^{n(H(\bar X)-\de)}
=H(\bar X)-\de > R,$$
where the inequality follows from property~\ref{sm1} of Proposition~\ref{aep}.
Hence, such $\{\hat X^n\}_{n=1}^\infty$ process,
uniformly distributed over its support, belongs to $\qc(R)$. Thus, we can lower-bound $E_{\text{PV}}^{(\theta)}(R)$ for channel
$\Wb=\{W^n\}_{n=1}^\infty$ and a given $\theta>1$ as follows
\begin{eqnarray*}
E_{\text{PV}}^{(\theta)}(R) &\triangleq&
\sup_{\boldsymbol X\in{\cal Q}(R)}\liminf_{n\rightarrow\infty}-\frac 1n\log
P_{X^nW^n}\bigg\{(x^n,y^n)\in\xc^n\times\yc^n:
\frac{1}{n}j_{X^nW^n}^{(\theta)}(x^n;y^n)\leq R \bigg\}
\\
&\ge& \liminf_{n\rightarrow\infty}
-\frac 1n\log P_{{\hat X}^nW^n}\bigg\{(x^n,y^n)\in\xc^n\times\yc^n: 
\frac 1nj_{{\hat X}^nW^n}^{(\theta)}(x^n;y^n)\leq R\bigg\}.
\end{eqnarray*}
For $n$ sufficiently large, we can write
\begin{eqnarray*} 
\lefteqn{j_{\hat X^nW^n}^{(\theta)}(x^n;y^n)}\\
&=&\log\displaystyle\frac{P^{\theta}_{Y^n|X^n}(y^n|x^n)}
{\sum_{\hat x^n\in\xc^n}P_{\hat X^n}(\hat x^n)
P^{\theta}_{Y^n|X^n}(y^n|\hat x^n)}\\
&=&\log\displaystyle\frac{P^{\theta}_{Y^n|X^n}(y^n|x^n)}
{\sum_{\hat x^n\in\fc_n(\delta|\bar X)}P_{\hat X^n}(\hat x^n)
P^{\theta}_{Y^n|X^n}(y^n|\hat x^n)+
\sum_{\hat x^n\not\in\fc_n(\delta|\bar X)}P_{\hat X^n}(\hat x^n)
P^{\theta}_{Y^n|X^n}(y^n|\hat x^n)}\\
&=&\log\displaystyle\frac{P^{\theta}_{Y^n|X^n}(y^n|x^n)}
{\sum_{\hat x^n\in\fc_n(\delta|\bar X)}P_{\hat X^n}(\hat x^n)
P^{\theta}_{Y^n|X^n}(y^n|\hat x^n)} \\
&\ge& \log\displaystyle\frac{P^{\theta}_{Y^n|X^n}(y^n|x^n)}
{\frac{e^{2n\delta}}{1-\delta}
\sum_{\hat x^n\in\fc_n(\delta|\bar X)}P_{\bar X^n}(\hat x^n)
P^{\theta}_{Y^n|X^n}(y^n|\hat x^n)} \\
&\ge&\log\displaystyle\frac{(1-\de)e^{-2n\delta}P^{\theta}_{Y^n|X^n}(y^n|x^n)}
{\sum_{\hat x^n\in\xc^n}P_{\bar X^n}(\hat x^n)
P^{\theta}_{Y^n|X^n}(y^n|\hat x^n)} \\
&=&\log(1-\delta)-2n\delta +j_{\bar X^nW^n}^{(\theta)}(x^n;y^n),
\end{eqnarray*}
where the first inequality follows from the lower bound in \eqref{aep-ineq}.
Accordingly,
\begin{eqnarray}
E_{\text{PV}}^{(\theta)}(R) &\ge&
\liminf_{n\rightarrow\infty}
-\frac 1n\log P_{{\hat X}^nW^n}\bigg\{(x^n,y^n)\in\xc^n\times\yc^n:
\frac 1nj_{{\hat X}^nW^n}^{(\theta)}(x^n;y^n)\leq R\bigg\} \nonumber \\
&\ge& \liminf_{n\rightarrow\infty}
-\frac 1n\log P_{{\hat X}^nW^n}\bigg\{(x^n,y^n)\in\xc^n\times\yc^n:
\frac 1n\log(1-\delta)-2\delta \nonumber \\
&& \hspace{3.0in} +\frac 1nj_{{\bar X}^nW^n}^{(\theta)}(x^n;y^n)\leq R\bigg\}
\nonumber  \\
&=& \liminf_{n\rightarrow\infty}
-\frac 1n\log P_{{\hat X}^nW^n}\bigg\{(x^n,y^n)\in\xc^n\times\yc^n:
\frac 1nj_{{\bar X}^nW^n}^{(\theta)}(x^n;y^n)  \nonumber \\
&& \hspace{3.0in} \leq R-\frac 1n\log(1-\delta)+2\delta\bigg\}.
\label{xx}
\end{eqnarray}
Observe that
\begin{equation*}
P_{\hat X^nW^n}(x^n,y^n)=P_{\hat X^n}(x^n)P_{Y^n|X^n}(y^n|x^n)
\leq\frac {e^{2n\delta}}{1-\delta}P_{\bar X^n}(x^n)P_{Y^n|X^n}(y^n|x^n),
\end{equation*}
where the inequality follows from \eqref{aep-ineq}. 
Then, we can further lower-bound the right-hand side term of \eqref{xx} to obtain
\begin{eqnarray*}
E_{\text{PV}}^{(\theta)}(R) &\ge&
\liminf_{n\rightarrow\infty}
-\frac 1n\log \bigg( \frac {e^{2n\delta}}{1-\delta} P_{\bar X^nW^n}
\bigg\{(x^n,y^n)\in\xc^n\times\yc^n: \nonumber \\
&& \hspace{2.0in} \frac 1nj_{{\bar X}^nW^n}^{(\theta)}(x^n;y^n) \leq 
R-\frac 1n\log(1-\delta)+2\delta\bigg\} \bigg) \nonumber \\
&\ge& \liminf_{n\rightarrow\infty} -\frac 1n\log P_{\bar X^nW^n}
\bigg\{(x^n,y^n)\in\xc^n\times\yc^n: \frac 1nj_{{\bar X}^nW^n}^{(\theta)}(x^n;y^n) 
\leq R+\gamma \bigg\} -2\delta, 
\end{eqnarray*}
where it suffices to take $\gamma>2\delta$ to have
$\gamma>-\frac 1n\log(1-\delta)+2\delta$ for
$n$ sufficiently large.

In summary, we have shown that for any channel $\Wb=\{W^n\}_{n=1}^\infty$, the upper
bound $E_{\text{PV}}^{(\theta)}(R)$ to its channel reliability function satisfies
\begin{eqnarray*}
E_{\text{PV}}^{(\theta)}(R) &\ge& 
\liminf_{n\rightarrow\infty} -\frac 1n\log P_{\bar X^nW^n}
\bigg\{(x^n,y^n)\in\xc^n\times\yc^n: \frac 1nj_{{\bar X}^nW^n}^{(\theta)}(x^n;y^n) 
\leq R+\gamma \bigg\} -2\delta 
\end{eqnarray*}
for $\theta\ge 1$ and any i.i.d.\ input process ${\bar X}$ with
$$\begin{cases}
H(\bar X) >R & \\
0<\delta<H(\bar X)-R & \\
\gamma > 2 \delta. & 
\end{cases}
$$
We next specialize the above lower bound for the case when channel 
$\Wb$ is memoryless.
For a memoryless channel with an i.i.d.~input, we have for $\rho<0$,
\begin{eqnarray*}
\lefteqn{P_{\bar X^nW^n}\left\{(x^n,y^n)\in\xc^n\times\yc^n:
\frac 1nj_{\bar X^nW^n}^{(\theta)}(x^n;y^n)\leq
 R+\gamma\right\}}\\
&=&
P_{\bar X^nW^n}\left\{(x^n,y^n)\in\xc^n\times\yc^n:\rho\sum_{i=1}^n \log\frac{P_
{Y|X}^{\theta}(y_i|x_i)}{\sum_{x'\in\xc}P_{\bar X}(x')P_{Y|X}^{\theta}(y_i|x')}\geq
 n\rho(R+\gamma)\right\}\\
&\leq&\left(e^{-\rho(R+\gamma)}
\left[\sum_{x\in\xc}\sum_{y\in\yc}P_{\bar X}(x)P_{Y|X}(y|x)e^{\rho \log\frac{P_{
Y|X}^{\theta}(y|x)}{\sum_{x'\in\xc}P_{\bar X}(x')P_{Y|X}^{\theta}(y|x')}}\right]
\right)^n\\
&=&\left(e^{-\rho(R+\gamma)}
\left[\sum_{x\in\xc}\sum_{y\in\yc}P_{\bar X}(x)P_{Y|X}(y|x)\left(\frac{P_{Y|X}^{
\theta}(y|x)}{\sum_{x'\in\xc}P_{\bar X}(x')P_{Y|X}^{\theta}(y|x')}\right)^\rho
\right]\right)^n\\
&=&\left(e^{-\rho(R+\gamma)}
\left[\sum_{x\in\xc}\sum_{y\in\yc}P_{\bar X}(x)\frac{P_{Y|X}^{1+\rho\theta}(y|x)}
{\left(\sum_{x'\in\xc}P_{\bar X}(x')P_{Y|X}^{\theta}(y|x')\right)^\rho}\right]
\right)^n,
\end{eqnarray*}
where the inequality follows from Markov's inequality.
Thus, for $\rho <0$, we have
\begin{eqnarray*}
E_{\text{PV}}^{(\theta)}(R) &\ge&
\liminf_{n\rightarrow\infty} -\frac 1n\log P_{\bar X^nW^n}
\bigg\{(x^n,y^n)\in\xc^n\times\yc^n: \frac 1nj_{{\bar X}^nW^n}^{(\theta)}(x^n;y^n)
\leq R+\gamma \bigg\} -2\delta \\
&\geq&\liminf_{n\rightarrow\infty}
-\frac 1n\log\left(e^{-\rho(R+\gamma)}
\left[\sum_{x\in\xc}\sum_{y\in\yc}P_{\bar X}(x)\frac{P_{Y|X}^{1+\rho\theta}(y|x)
}{\left(\sum_{x'\in\xc}P_{\bar X}(x')P_{Y|X}^{\theta}(y|x')\right)^\rho}\right]
\right)^n-2\delta\\
&\geq&\liminf_{n\rightarrow\infty}
\left({\rho(R+\gamma)}-\log
\left[\sum_{x\in\xc}\sum_{y\in\yc}P_{\bar X}(x)\frac{P_{Y|X}^{1+\rho\theta}(y|x)
}{\left(\sum_{x'\in\xc}P_{\bar X}(x')P_{Y|X}^{\theta}(y|x')\right)^\rho}\right]
\right)-2\delta\\
&=&{\rho(R+\gamma)}-\log
\left[\sum_{x\in\xc}\sum_{y\in\yc}P_{\bar X}(x)\frac{P_{Y|X}^{1+\rho\theta}(y|x)
}{\left(\sum_{x'\in\yc}P_{\bar X}(x')P_{Y|X}^{\theta}(y|x')\right)^\rho}\right]-2\delta.
\end{eqnarray*}
Since $\rho<0$, $\gamma$ should be made as small as possible. But as $\gamma>2\delta$,
it should thus approach $2\delta$ to obtain
\begin{eqnarray*}
E_{\text{PV}}^{(\theta)}(R) &\ge&
\rho R -\log
\left[\sum_{x\in\xc}\sum_{y\in\yc}P_{\bar X}(x)\frac{P_{Y|X}^{1+\rho\theta}(y|x)
}{\left(\sum_{x'\in\yc}P_{\bar X}(x')P_{Y|X}^{\theta}(y|x')\right)^\rho}\right]
-2(1-\rho)\delta. 
\end{eqnarray*}
Taking $\delta \downarrow 0$ 
yield the following lower bound 
to $E_{\text{PV}}^{(\theta)}(R)$ for a memoryless channel
\begin{eqnarray}
E_{\text{PV}}^{(\theta)}(R) &\ge& \sup_{P_{\bar X}:H(\bar X) >R} \sup_{\rho <0} \left\{
\rho R -\log
\left[\sum_{x\in\xc}\sum_{y\in\yc}P_{\bar X}(x)\frac{P_{Y|X}^{1+\rho\theta}(y|x)
}{\left(\sum_{x'\in\yc}P_{\bar X}(x')P_{Y|X}^{\theta}(y|x')\right)^\rho}\right]
\right\} \nonumber \\
&\triangleq& F(R,\theta)
\label{F-R-theta}
\end{eqnarray}
for $\theta\ge 1$.


\bigskip

\begin{thebibliography}{10}

\bibitem{pv02} F.~Alajaji, P.-N.~Chen and Z.~Rached, ``A note on the Poor-Verd\'{u} upper
bound for the channel reliability function,'' {\em IEEE Trans.\ Inform.\ Theory}, vol.~48,
no.~1, pp.~309--313, Jan.~2002.

\bibitem{barg} A.~Barg and A.~McGregor, ``Distance distribution of binary codes
and the error probability of decoding,'' {\em IEEE Trans.\ Inform.\ Theory}, 
vol.~51, pp.~4237--4246, Dec.~2005.

\bibitem{ben-haim} Y.~Ben Haim and S.~Litsyn, ``Improved upper bounds on the 
reliability function of the Gaussian channel,'' 
{\em IEEE Trans.\ Inform.\ Theory}, vol.~54., no.~1, pp.~5--12, Jan.~2008.


\bibitem{billingsley} P.\ Billingsley, {\em Probability and Measure}, Second Edition, 
Wiley, NY, 1986.

\bibitem{blahut} R.\ Blahut, {\em Principles and 
Practice of Information Theory},
Addison Wesley, MA, 1988.

\bibitem{Bucklewbook} J.\ A.\ Bucklew, {\em Large Deviation Techniques in
Decision, Simulation, and Estimation}, Wiley, NY, 1990.



\bibitem{cover-thomas} 
 T.~M.~Cover and J.A.~Thomas, {\em Elements of
Information Theory}, New York: Wiley, 2nd Ed., 2006.

\bibitem{csiszar} I.\ Csisz\'{a}r and J.\ K\"{o}rner, 
{\em Information Theory: Coding Theorems for Discrete 
Memoryless Systems}, Academic Press, NY, 1981. 


\bibitem{gallager} R.\ G.\ Gallager, {\em Information Theory
and Reliable Communication}, Wiley, NY, 1968.

\bibitem{Hanbook} T.\ S.\ Han, {\em Information-Spectrum Methods in
Information Theory}, Springer, 2003.

\bibitem{miller-thomas} J.~H.~Miller and J.~B.~Thomas, ``Detectors
for discrete-time signals in non-Gaussian noise,''
{\em IEEE Trans.\ Inform.\ Theory}, vol.~18, no.~2, pp.~241--250, Mar.~1972.


\bibitem{PV95} H.\ V.\ Poor and S.\ Verd\'{u},
``A lower bound on the probability 
of error in multi-hypothesis testing,''
{\em IEEE Trans.\ Inform.\ Theory}, 
vol.~41, no.~6, pp.~1992--1994, Nov.~1995.

\bibitem{shannon57} C.~E.~Shannon, ``Certain results in coding theory
for noisy channels,'' {\it Inform.\ Contr.}, vol.~1, pp.~6--25, Sep.~1957. 

\bibitem{VH94} S.\ Verd\'{u} and T.\ S.\ Han,
``A general formula for channel capacity,''
{\em IEEE Trans.\ Inform.\ Theory}, 
vol.~40, no.~4, pp.~1147--1157, July 1994.

\bibitem{viterbi} A.\ J.\ Viterbi and J.\ K.\ Omura, {\it Principles 
of Digital Communication and Coding}, McGraw-Hill, NY, 1979.

\end{thebibliography}
\end{document}